\documentclass[a4paper,11pt]{article}
\pdfoutput=1

\usepackage{tcs}

\begin{document}
\title{Metric Embeddings Beyond Bi-Lipschitz Distortion \\ via Sherali-Adams}
%\title{A quasi-polynomial time algorithm for \\
%Multi-Dimensional Scaling via LP hierarchies }

\author{
Ainesh Bakshi \\
% \texttt{ainesh@mit.edu} \\
MIT
\and
Vincent Cohen-Addad \\
% \texttt{} \\
Google Research 
\and
Samuel B. Hopkins \\
% \texttt{samhop@mit.edu} \\
MIT
\and
Rajesh Jayaram \\
% \texttt{} \\
Google Research
\and 
Silvio Lattanzi\\
% \texttt{} \\
Google Research
}

\maketitle
\begin{abstract}
%\Snote{TODO}
%\begin{center}
%\includegraphics[scale = 0.2]{frog.jpg}
%\end{center}

Metric embeddings are a widely used method in algorithm design, where generally a ``complex'' metric is embedded into a simpler, lower-dimensional one. Historically, the theoretical computer science community has 
focused on \emph{bi-Lipschitz} embeddings, which guarantee that every pairwise distance is approximately preserved. 
In contrast, alternative embedding objectives that avoid bi-Lipschitz distortion are commonly used in practice to map points to lower dimensions, yet these approaches have received comparatively less study in theory.

In this paper, we focus on one such objective, Multi-dimensional Scaling (MDS), which embeds an $n$-point metric into low-dimensional Euclidean space. MDS is widely used as a data visualization tool in the social and biological sciences, statistics, and machine learning.
Given a set of non-negative dissimilarities $\{d_{i,j}\}_{i , j \in [n]}$ over $n$ points (which may or may not form a metric), the goal is to find an embedding $\{x_1,\dots,x_n\} \subset \R^k$ that minimizes
 \[ \OPT = \min_{x}  \expecf{i,j \in [n]}{ \left(1-\frac{\|x_i - x_j\|}{d_{i,j}}\right)^2  }.\] 

Despite its popularity, our theoretical understanding of MDS is extremely limited. Recently,  Demaine, Hesterberg, Koehler, Lynch, and Urschel~\cite{demaine2021multidimensional} gave the first approximation algorithm with provable guarantees for this objective, which achieves an embedding in constant dimensional Euclidean space with cost $\OPT +\eps$ in $n^2  \cdot 2^{ \poly(\Delta/\eps) }$ time, where $\Delta$ is the aspect ratio of the input dissimilarities. For metrics that admit low-cost embeddings, the aspect ratio $\Delta$ scales polynomially in $n$.
In this work, we give the first approximation algorithm for MDS with quasi-polynomial dependency on $\Delta$: for constant dimensional Euclidean space, we achieve a solution with cost $\tilde{\mathcal{O}}(\log \Delta ) \cdot \OPT^{ \Omega(1) } + \eps$ in time $n^{ \mathcal{O}(1)} \cdot 2^{  \poly((\log(\Delta)/\eps))  }$.\footnote{Throughout, we use the notation $\tilde{\mathcal{O}}(\cdot)$ to suppress polylogarithmic factors in all the supplied arguments. } Our algorithms are based on a novel \emph{geometry-aware} analysis of a conditional rounding of the Sherali-Adams LP Hierarchy, allowing us to avoid exponential dependency on the aspect ratio of the input that would typically result from this rounding.

\end{abstract}

%%%% MDS Haiku %%%%%%%
%%% In MDS, data forms a map,
% Patterns in space, a visual trap.
% Insights in low-dimensions wrap. %%%%

\thispagestyle{empty}

\newpage
\thispagestyle{empty}

\tableofcontents

\newpage

\setcounter{page}{1}

\section{Introduction}

Metric embeddings have long been a central topic in theoretical computer science and are a foundational tool for understanding and simplifying complex structures in high-dimensional spaces. However, most of the literature on metric embeddings focuses on \emph{bi-Lipschitz} embeddings, which preserve all distances up to a multiplicative distortion factor $L \geq 1$.
That is, bi-Lipschitz embeddings satisfy $d_{i,j} \leq \|x_i - x_j\| \leq L \cdot d_{i,j}$ for all pairs $i,j$.
While the theory of bi-Lipschitz embeddings is extremely rich (see~\cite{indyk20178} and references therein), strong computational hardness results are known for the problem of minimizing the distortion $L$ over bi-Lipschitz embeddings of a given metric into $\R^k$ for fixed $k$ \cite{matouvsek2010inapproximability}.
Specifically, it is NP-hard approximate the best distortion $L$ for a given $n$-point metric $d$ to any factor better than $n^{\Theta(1/k)}$, and furthermore, this approximation factor can be obtained in polynomial time by first embedding $\{d_{i,j}\}_{i,j \in [n]}$ into high-dimensional Euclidean space using an approach of Bourgain \cite{bourgain1985lipschitz} and then taking a random linear projection.
If $\{d_{i,j}\}_{i,j \in [n]}$ is already a (high-dimensional) Euclidean metric, this approach is no better than projection to a random $k$-dimensional subspace.

Bi-Lipschitz-ness is, therefore, a poor yardstick for measuring the performance of efficient metric embedding algorithms, since a random embedding is essentially as bi-Lipschitz as the best embedding we can hope to obtain in polynomial time. Further, the worst-case distortion between any pair of points is not the only possible measure of the quality of an embedding. In fact, many popular dimensionality reduction objectives in practice focus on alternative embedding objectives, adopting a smoother notion of distortion. In this paper, we demonstrate that one such popular family of embedding objectives admits non-trivial approximation algorithms, making them a more fruitful testbed for algorithm design. 

Specifically, we study the \emph{multi-dimensional scaling} (MDS) objective. 
MDS dates at least to the 1930s \cite{richardson1938multidimensional}, and was popularized in the 1950s and 1960s psychometrics literature by influential works of Torgerson \cite{torgerson1952multidimensional}, Shephard \cite{shepard1962analysis-1,shepard1962analysis-2}, and Kruskal \cite{kruskal1964multidimensional,kruskal1964nonmetric}.
MDS has since become a bread-and-butter technique in applied statistics -- it is the subject of several books \cite{kruskal1978multidimensional,cox2000multidimensional,borg2005modern,borg2012applied} and makes frequent appearances in statistics textbooks \cite{davison2000multidimensional,handbook-of-statistics-MDS,johnson2002applied}.
Heuristic algorithms for MDS are built into several popular programming languages and packages, including Matlab, R, and scikit-learn \cite{MathWorks,de2011multidimensional,pedregosa2011scikit}. A Google Scholar search for ``multidimensional scaling'' turns up more than half a million results.

Formally, given an $n$-point metric $\{d_{ij}\}_{i,j \in [n]}$, the Kamada-Kawai cost-function for MDS, henceforth referred to as \emph{KK}, is defined as follows:
\begin{align*}
    \OPT = \min_{x_1,\ldots,x_n \in \R^k}  \E_{i , j \sim [n] } \Paren{ 1 - \frac{ \norm{ x_i - x_j} }{d_{i,j}} }^2 \, .
    \tag{Kamada-Kawai (KK) }
\end{align*}
Here, $\E_{i,j \sim [n]}$ denotes a uniformly random pair $i,j$ drawn from $[n] \times [n]$.
KK measures the mean-squared distance between $1$ and the multiplicative distortion $\|x_i - x_j\| / d_{i,j}$ experienced by the pair $i,j$.
For normalization, note that placing all $x_i$ at the same point gives objective value $1$, so for any input $\{d_{ij}\}_{i,j \in [n]}$, the optimum value $\OPT$ is in $[0,1]$.

KK is especially popular as a tool for force-directed graph drawing in the social and medical sciences.
For instance, the KK objective has been used to visualize the spread of obesity and the dynamics of smoking in social networks~\cite{christakis2007spread,christakis2008collective}, for mapping the structural core of the human cerebral cortex~\cite{hagmann2008mapping}, and identifying intellectual turning points in theoretical physics~\cite{chen2004searching}.
In practice, the embedding $x_1,\ldots,x_n$ is typically obtained via gradient descent or other local search methods applied to the objective function.
These objective functions are nonconvex, and we are aware of no provable guarantees for such local search procedures.
%the goal of \emph{multi-dimensional scaling} (MDS) is to find a low-dimensional embedding $x_1,\ldots,x_n \in \R^k$ such that $d_{ij} \approx \|x_i - x_j\|$.
%MDS is useful for many downstream tasks: feature extraction for machine learning, dimension reduction, data visualization (especially if $k = 2,3$), graph drawing (if $d_{ij}$ a shortest-path metric), clustering, and more.
We are aware of few works giving algorithms with provable guarantees for any MDS cost function, including KK.
Before we turn to these and to our results, we contrast MDS with some familiar metric embedding problems studied in theoretical computer science.

\paragraph{MDS versus Principal Component Analysis.}
Principal component analysis (PCA) is a well-studied approach to dimension reduction of a high-dimensional Euclidean metric.
It can be phrased via a similar-looking optimization problem as KK, applied to inner products rather than distances.
If $y_1,\ldots,y_n \in \R^N$, then
\[
\textsf{PCA} = \min_{x_1,\ldots,x_n \in \R^k} \E_{i,j \in [n]} (\langle {y_i,y_j} \rangle - \langle {x_i,x_j} \rangle )^2 \, .
\]
It is well known that the optimizing $x_1,\ldots,x_n$ are given by $\Pi_k y_1,\ldots, \Pi_k y_n$, where $\Pi_k$ is the projector to the top $k$ eigenvectors of the covariance matrix of $y_1,\ldots,y_n$.
In particular, the optimal embedding is a linear projection of $y_1,\ldots,y_n$.

Since inner products are related to distances via $\langle{y_i,y_j} \rangle = \tfrac 12 (\|y_i - y_j\|^2 - \|y_i\|^2 - \|y_j\|^2)$, is it tempting to imagine that the least-squares optimization problems defining PCA and MDS are equivalent or close to it.
Indeed, a popular MDS heuristic, sometimes called \emph{classical multidimensional scaling}, applies PCA to the matrix $(I - \tfrac 1 n 1 1^\top) D^{(2)} (I - \tfrac 1 n 11^\top)$, where $D^{(2)}$ is the entrywise square of the distance matrix $D$ and $1$ denotes the all-$1$s vector -- if $\{d_{ij}\}_{i,j \in [n]}$ is a Euclidean metric $y_1,\ldots,y_n \in \R^N$, then this is equivalent to applying PCA to $y_1,\ldots,y_n$.

To the best of our knowledge, no provable approximation guarantees for KK (or any other MDS cost function) are known for this algorithm.
In general, we do not expect the PCA-optimal embedding to be related to the KK-optimal embedding, and unlike PCA, the KK-optimal embedding can be a \emph{nonlinear} mapping of $y_1,\ldots,y_n$ into $\R^k$. To demonstrate this power, we construct an explicit example where the KK-optimal embedding recovers interesting cluster structure which is lost by PCA in Section \ref{sec:MDSvPCA}. Moreover, in this example, the solution produced by PCA has an arbitrarily worse MDS cost than the optimal MDS solution, despite being close to low-dimensional initially. Thus, even for pointsets that have a small optimal MDS cost to begin with, the quality resulting from using PCA as an algorithm for the MDS objective can be arbitrarily poor.

%\Snote{Can we get a simple example where PCA and MDS obviously give different optimal embeddings?}

\paragraph{State of the Art: Large Embedding Dimension or Exponential Dependence on Aspect Ratio.}
Having argued the importance of designing algorithms with provable guarantees for MDS, both because of its widespread use in practice and as a new approach to the algorithmic theory of metric embeddings, we turn to the algorithmic state of the art.

We are aware of only two works giving algorithms with provable approximation guarantees for MDS.
Bartal, Fandina, and Neiman~\cite{bartal2019dimensionality} observe that many MDS objectives become convex (semidefinite) programs if we ask to embed in $n$ dimensions -- this is true in particular for $\OPT$.
A simple Johnson-Lindenstrauss-style analysis shows that randomly projecting an $n$-dimensional KK solution, obtained via semidefinite programming, to $k$ dimensions incurs an additive $O(1/k)$ loss in objective value.
Thus, we can obtain a solution with cost $\OPT + O(1/k)$ in polynomial time.
But, as we have already discussed, for small $k$ (e.g. $k = 1,2,3$), random projection yields poor embeddings, and consequently this algorithm lacks nontrivial guarantees for such small $k$.

Demaine, Hesterberg, Koehler, Lynch, and Urschel give an algorithm that retains provable guarantees for small $k$, at the expense of much slower running time~\cite{demaine2021multidimensional}.
They associate an important parameter to the input metric: Let
\[
\Delta = \frac{\max_{i,j \in [n]} d_{ij}}{\min_{i,j \in [n]} d_{ij}}
\]
measure the \emph{aspect ratio} of the metric.
\cite{demaine2021multidimensional} 
obtains an efficient algorithm when the aspect ratio ($\Delta$) is small. Their result can be summarized as follows: 
% \ainesh{we need to edit this in light of the Bartal paper}

% {\color{blue} Suggestion: The lack of study of the objective
% in the algorithm design and theoretical computer science 
% community presumably comes from the fact that the objective 
% is not easily maneuverable: any multiplicative distortion 
% of a distance compared to the optimum one results in an 
% additive increment of the cost. Hence the result of Demaine, Hesterberg, Koehler, Lynch, and Urschel goes as follows}
% \ainesh{redo above paragraph...}
% \Snote{what is a sentence in blue supposed to be buying us? seems like total speculation. I would omit it.}
\begin{theorem}[Approximation schemes scaling exponentially in aspect ratio~\cite{demaine2021multidimensional}]
\label{thm:demaine}
For every $k > 0$ and $p \geq 0$ there is an algorithm with running time $n^2 \cdot \exp ( ( 
\Delta / \eps)^{O(1)})$ which outputs an embedding with KK cost $\OPT + \eps$.
\end{theorem}
% The algorithm of \cite{demaine2021multidimensional} is polynomial-time for $\Delta$ as large as $\approx (\log n)^{1/4}$.
% On the plus side, for large-enough $\Delta$, \cite{demaine2021multidimensional} shows that the additive PTAS becomes a multiplicative PTAS, by lower-bounding $\OPT$. 

How should we think of the aspect ratio parameter $\Delta$?
If $\{d_{i,j}\}_{i,j \in [n]}$ represent, say, a $1$-dimensional Euclidean metric, then $\Delta = \Omega(n)$.
In general, for metrics with good low-dimensional embeddings, \cite{demaine2021multidimensional}'s algorithm has exponential running time,\footnote{We thank Bobby Kleinberg for pointing out this perspective to us.}
suggesting the question:

\begin{quote}
\begin{center}
   \emph{Is there an approximation algorithm for Kamada-Kawai which is efficient on inputs with large aspect ratio?}
\end{center}
\end{quote}
We note that inputs with large aspect ratio are of particular interest since they include instances that admit small-objective-value embeddings in low dimensions.

\subsection{Our Results}
We give an almost-exponential improvement on the running time compared to \cite{demaine2021multidimensional} with respect to $\Delta$, at the cost of a somewhat worse approximation factor.
As far as we know, we give the first nontrivial approximation algorithm for MDS, which remains polynomial-time for super-poly-logarithmic values of $\Delta$.
\begin{theorem}[Main theorem, qualitative version]
For every $k \geq 1$ there is an algorithm with running time $n^{O(1)} \cdot \exp((\log \Delta/\eps)^{O(1)})$ which outputs an embeddding with KK cost $$\tilde{O}(\log \Delta) \cdot (\OPT)^{\Omega(1)} + \eps \, .$$
\end{theorem}
We refer the reader to Remark~\ref{rem:k-dependence} for discussion of the dependence of the hidden constants on $k$.
Our result does not actually require that $\{d_{i,j}\}_{i,j \in [n]}$ form a metric -- it is enough for each to be a non-negative number in $[1,\Delta]$
(this is also true for Theorem~\ref{thm:demaine}).

\begin{remark}[Comparison with~\cite{demaine2021multidimensional}]
Consider the illustrative toy example where the input distances correspond to a small perturbation of the line metric on $n$ points. Here, the $\Set{d_{i,j}}_{i,j\in[n]}$ correspond to scalars $\Set{ x_i^* }_{i \in [n]}$ such that $d_{i,j} = (1\pm \zeta) \abs{ x_i^* - x_j^* } $, where the perturbation is bounded, i.e. $\abs{\zeta} \leq 1/\poly(n)$.
We note that such instances are of particular interest, since they are precisely the ones that admit good low-dimensional embeddings.

It is easy to see that $\Delta = \Omega(n)$ and when the target dimension $k=1$,  the KK objective has value $\OPT = 1/\poly(n)$.  Then, given an $1/\poly(n)< \eps<1$, \cite{demaine2021multidimensional} runs in time $2^{\poly(n/\eps)}$ and outputs an embedding with cost $\eps$. In contrast, our algorithm runs in time $n^{\poly(\log(n),1/\eps)}$ and also outputs an embedding with cost $\eps$. 
\end{remark}

% While our objective cost is slightly worse than~\cite{demaine2021multidimensional}, in the regime of interest, where the input admits a good low-dimensional embedding we expect the additive error to dominate. 

Our approach is to round a Sherali-Adams linear programming relaxation of the KK objective function.\footnote{
We could also use a more powerful Sum-of-Squares semidefinite programming relaxation, with no change to asymptotic running times or analysis, but we do not require it.}
We use $((\log \Delta)/\eps)^{O(k)}$ levels of the hierarchy, leading to an LP with $\approx n^{((\log \Delta)/\eps)^{O(k)}}$ variables and constraints.
(We can sparsify this LP to obtain our final running time.)

The conditioning-based rounding scheme we consider is a natural one -- it goes back at least to the \emph{global correlation rounding} algorithm of \cite{BRS} for dense constraint satisfaction problems.
In fact, a na\"ive analysis of the same rounding scheme, applied instead to $(\Delta / \eps)^{O(1)}$ levels of Sherali-Adams, is one way to obtain the result in~\cite{demaine2021multidimensional}.
In their proof of Theorem~\ref{thm:demaine}, the authors of \cite{demaine2021multidimensional} take a different perspective, directly reducing MDS to a dense constraint satisfaction problem.

Our main contribution is a new \emph{geometry-aware} analysis of conditioning-based Sherali-Adams rounding, which we describe in Section~\ref{sec:techniques}.
More broadly, we believe that new approaches to analyzing convex relaxations are a promising avenue for new metric embedding algorithms.
We outline a number of open problems in this direction in Section~\ref{sec:open-problems}.

\subsection{Related Work}
\label{sec:related-work}
% \Snote{TODO: things to cite -- t-sne literature, the t-sne paper by pravesh and sanjeev}
% \paragraph{Provable Approximations for Average-Case Distortion Objectives.}

Multi-dimensional scaling has been studied since at least the 1950s, with too vast a literature to survey properly here.
We refer the reader to the survey~\cite{young2013multidimensional} and book~\cite{cox2000multidimensional} for review
of its history and applications.
See also the recent book~\cite{agrawal2021minimum} for a perspective from applied optimization.
We discuss here some related works from the theoretical computer science literature.

\vspace{0.1in}
\noindent\textbf{t-SNE and Friends.}
Another popular non-linear dimension reduction technique/objective function, closely related to MDS, is t-SNE ~\cite{van2008visualizing,hinton2002stochastic}.
%is a
%related objective where the goal is to embed a metric space into $\R^2$ while essentially minimizing the Kullback-Leibler divergence of the input distance distributions vs the output distance distributions.
%\Snote{what are "input distance distributions"? i don't think that this description of t-sne is helping anyone who doesn't already know what t-sne is...}
Notably, Arora, Hu, and Kothari~\cite{DBLP:conf/colt/AroraHK18} showed
that the gradient-descent algorithm used in practice for t-SNE performs well
if the input exhibits some clusterability properties.
% \ainesh{i don't think this is in a setting that is related to ours.}
% \Snote{agreed, but it's kind of the most recent attempt we know of from tcs people to tackle a spiritually similar problem, so it seems worth citing.}\silvio{I also think that we should cite this work, it is known in this space.}
Other non-linear dimension reduction techniques used in practice include the Isomap embedding~\cite{balasubramanian2002isomap}, or spectral embedding
methods~\cite{von2007tutorial}. 
%We are not aware of any provable approximation
%algorithm for these approaches. 
%Further investigating
% the complexity of these problems and deriving provable approximation
% guarantees is a fascinating problem.

\vspace{0.1in}
\noindent\textbf{Approximation Algorithms via LP Hierarchies.}
% \ainesh{write about Sparsest cut being the last big hierachy based result in metric embeddings.}
% \vincent{Is sparsest cut hierarchy based? It is just SDP based, isn't it?}
There is a vast literature on approximation algorithms using LP hierarchies. Notable examples include works on scheduling with
communication delay~\cite{davies2020scheduling}, the matching polytope~\cite{DBLP:conf/stoc/MathieuS09}, numerous works on 
CSPs and Dense-CSPs~\cite{BRS,DBLP:conf/focs/GuruswamiS12}, Bin Packing~\cite{karlin2011integrality}, and Correlation Clustering~\cite{DBLP:conf/focs/Cohen-AddadLN22}.

\vspace{0.1in}
\noindent\textbf{Approximation Algorithms for Metric Embeddings with
Additive or Multiplicative Distortion}
Metric embeddings have been studied extensively for 30+ years in theoretical computer science, with extensive applications.
See e.g. references in the recent paper \cite{sidiropoulos2017metric}.

Strong inapproximability results~\cite{DBLP:conf/focs/MatousekS08} for minimizing bi-Lipschitz distortion led to some efforts to design approximation algorithms for relaxed measures of distortion, among them MDS.
We are not aware of works in this vein whose techniques apply directly to the Kamada-Kawai problem we consider here, so we defer further discussion to \cref{sec:additional-related-work}.

One additional influential work we highlight, slightly predating the inapproximability results, is \cite{abraham2005metric}, which gives a ``beacon-based'' algorithm for embedding metrics with small distortion on all but a small fraction of distances; this is another approach to relaxing bi-Lipschitzness in pursuit of better algorithmic guarantees.

\section{Techniques}
\label{sec:techniques}
We now provide an overview of our techniques.
We adopt the convention $1 \leq d_{i,j} \leq \Delta$ for all $i,j$.
This overview is meant to aid the reader in understanding the main ideas in our arguments, so the quantitative claims often ignore constant or logarithmic factors.
See Sections~\ref{sec:algo-and-analysis} and beyond for rigorous arguments.

\subsection{Our Algorithm}

\paragraph{Discretization.}
We begin by discretizing the KK problem. 
This makes it easier to formulate the Sherali-Adams LP hierarchy for KK. %, although we believe that our results could be obtained by formulating an LP directly without discretizing first.
Specifically, we show (\cref{lem:aspect-ratio}) that for any $k$ and input $\{d_{i,j}\}_{i,j \in [n]}$, the optimal embedding among those which place $x_1,\ldots,x_n$ at distinct points in a $(\Delta / \eps)^{O(k)}$-size net of $\R^k$ has cost at most $\eps$ more than the optimal embedding into $\R^k$. In particular, this requires us to demonstrate that there is an $\eps$-approximate optimal embedding which has aspect ratio at most $O(\Delta / \eps)$; we do this by proving that any optimal embedding can be projected onto a ball of radius $O(\Delta/\eps)$  without significantly increasing the cost. 
Since our approximation guarantees anyway lose an additive $\eps$, we can restrict attention to embeddings $x_1,\ldots,x_n$ into a discrete subset of $\R^k$ with $(\Delta / \eps)^{O(k)}$ points.
In this overview, we will not distinguish carefully between discretized and non-discretized versions of $\OPT$.

\paragraph{Sherali-Adams Hierarchy.}
For an integer $t \geq 1$, the level-$t$ Sherali-Adams linear programming (LP) relaxation for (discretized) KK is an LP with $(n \Delta / \eps)^{O(kt)}$ variables and constraints, whose solutions are \emph{Sherali-Adams pseudoexpectations} (or ``pseudoexpectations'' for short), denoted $\pE$.

A level-$t$ pseudoexpectation is a collection of ${n \choose t}$ ``local distributions'', one for each $S \subseteq [n]$ of size $|S| = t$.
The local distribution $\mu_S$ is a joint probability distribution over assignments of $\{x_i\}_{i \in S}$ to elements of $\R^k$.
The constraints of the Sherali-Adams LP ensure that local distributions $\mu_S, \mu_T$ have the same marginal distribution on $S \cap T$.
For purposes of designing a rounding scheme for the LP, we often pretend that $\pE$ defines an honest-to-goodness joint probability distribution over assignments of $x_1,\ldots,x_n$ to $\R^k$, although we can only look at this distribution ``locally''.
This mentality dates to \cite{BRS,raghavendra2012approximating} and has since been used extensively to design algorithms by rounding Sherali-Adams LPs and Sum-of-Squares semidefinite programs.

\emph{Pseudoexpectations and variance:} For any real-valued function $f$ depending on a subset $S$ of at most $t$ out of $x_1,\ldots,x_n$, $f$ has a well-defined ``pseudo-expected'' value, denoted $\pE f = \E_{x_S \sim \mu_S} f(x_S)$.
We can extend $\pE$ to be a linear operator which assigns a real value to any linear combination of such $t$-local functions.
We can similarly measure the variance of any $t$-local function via $\pE ( f - \pE f)^2$, since $(f - \pE f)^2$ is itself a $t$-local function.

\emph{Conditioning:} For any level-$t$ pseudoexpectation $\pE$, index $i \in [n]$, and $z \in \R^k$ such that $\Pr(x_i = z) > 0$,\footnote{We can define this probability via $\pE 1_{x_i = z}$.} we can define a level-$(t-1)$ conditional pseudoexpectation $\pE[ \cdot \, | \, x_i = z]$ by defining $\mu_{S}$ to be $\mu_{S \cup \{x_i\}}$ conditioned on $x_i = z$.

\paragraph{Rounding by Conditioning.}
Our algorithm first solves the Sherali-Adams LP to find a level $t$ pseudoexpectation minimizing 
\begin{align}
\min_{\pE} \pE \E_{i,j} \left ( 1- \frac{\|x_i - x_j\|}{d_{i,j}} \right )^2 \leq \OPT \, , \label{eq:tech-0}
\end{align}
where the inequality holds because the LP is a relaxation of the minimization problem defining $\OPT$.
We ignore some technicalities in this overview; in particular, our purpose here we will set $t \approx (\log \Delta) / \eps^{O(k)}$.
Later, in Section \ref{sec:algo-and-analysis}, with the goal of optimizing the hidden constants in the approximation guarantee, we will actually set $t = ((\log \Delta) / \eps)^{O(k)}$, although noting that the approach given in this overview is valid but would give a different runtime-approximation trade-off --- see the discussion after Theorem~\ref{thm:efficient-algo-mds}.

To round $\pE$ we choose a random set $S = \{i_1,\ldots,i_{t-2}\} \subset [n]$ of size $|S| = t-2$, where each of $i_1,\ldots,i_{t-2}$ is drawn uniformly and independently at random.
Then, we draw a sample $z_S \sim \mu_S$ from the local distribution for $S$ and output the expected locations of each $x_i$ conditioned on $x_S = z_S$; i.e., we embed the $i$-th point at $\pE[ x_i \, | \, x_S = z_S] \in \R^k$.

\begin{mdframed}
  \begin{algorithm}[Efficient Kamada-Kawai (informal)]
    \label{algo:efficient-algo-intro}\mbox{}
    \begin{description}
    \item[Input:] Nonnegative numbers $\calD =\Set{ d_{i,j} }_{i,j \in [n]}$, target dimension $k \in \mathbb{N}$, target accuracy $0<\eps<1$.
    
    \item[Operation:]\mbox{}
    
\begin{compactenum}
    \item Let $t = (\log \Delta)/\eps^{O(k)}$.
    Let $\pE$ be a degree $t$ Sherali-Adams pseudoexpectation over a $(\Delta/\eps)^{O(k)}$ size net of $\R^k$, optimizing
    \[
    \min_{\pE} \hspace{0.1in} \pE  \E_{i,j \sim [n]} \Paren{ 1 - \frac{\|x_i - x_j\|}{d_{ij}}}^2 \, .
    \]
    \item Let $S\subset [n]$ be a random subset of size $t-2$ and let $z_S$ be a draw from the local distribution $\mu_S$ induced by $\pE$.
    
\end{compactenum} 
    \item[Output:]  The embedding $\Set{ \pexpecf{}{ x_i \, | \, x_S = z_S} }_{i \in [n]}$.
    \end{description}
  \end{algorithm}
\end{mdframed}

\subsection{Our Analysis}
Our goal is now to bound the expected objective value of the rounded solution
\[
\E_{S} \E_{z_S \sim \mu_S} \E_{i,j} \left ( 1- \frac{\|\pE [x_i \, | \, x_S = z_S ] - \pE [x_j \, | \, x_S = z_S ]\|}{d_{i,j}} \right )^2  \, 
\]
in terms of the LP value \eqref{eq:tech-0} and $\eps > 0$.

To get a sense of the challenge we face, let us see what could go wrong if we output $\pE x_i$ as the embedded location of point $i$, without first doing the conditioning step.
Even if $\OPT$ is $0$, meaning that $\{d_{ij}\}_{i,j \in [n]}$ has a perfect embedding into $\R^k$, the pseudoexpectation $\pE$ could have $\pE x_i = 0$ for all $i$, yielding the worst-possible rounded cost of $1$.
For instance, this could arise by first taking a perfect embedding $y_1,\ldots,y_n \in \R^k \setminus \{0\}$ of the input (i.e. $\|y_i-y_j\|_2 = d_{ij}$ for all $i,j$), and obtaining a probability distribution over embeddings which takes the values $(x_1,\dots,x_n) = (y_1,\ldots,y_n)$ with probability $1/2$ and $(x_1,\dots,x_n) = (-y_1,\ldots,-y_n)$ with probability $1/2$.

Let us see that conditioning addresses the issue in the context of this example.
If we sample, say, $z_1 = y_1$ and condition on $x_1 = z_1$, the conditional distribution of $x_i$ places all its mass $y_i$, giving the perfect embedding we are looking for.

In general, of course, we cannot assume that $\OPT = 0$ or that the pseudoexpectation $\pE$ obtained via \eqref{eq:tech-0} is an actual distribution over embeddings, much less one obtained by taking a single fixed embedding and randomly negating it.
Instead, we develop a more robust version of the above argument.
Although after conditioning on $x_S = z_S$ the distributions of the remaining $x_i$ will not be supported on a single point as they were in the above example, we will be able to show that for most $i$ the conditional distribution has small variance.
Next, we show that such a variance bound is enough to bound the rounded cost.

\paragraph{Small Variance Implies Small Rounded Cost.}
Consider the contribution of a pair $(i,j)$ to the rounded objective.
Suppose we are lucky, in the following sense:
after conditioning on $x_S = z_S$, the variance of both $x_i$ and $x_j$ is small compared to $d_{i,j}^2$.
That is, suppose:
\begin{align}
\E_{z_S \sim \mu_S} \pE [ \| x_i - \pE[ x_i \, | \, x_S = z_S ] \|^2 \, \big | \, x_S = z_S ], \, \, \E_{z_S \sim \mu_S} \pE [ \| x_j - \pE[ x_j \, | \, x_S = z_S ] \|^2 \, \big | \, x_S = z_S ]  \leq \eps \cdot d_{i,j}^2 \, .
\label{eq:tech-1.5}
\end{align}
Now, for any jointly-distributed $\R^k$-valued random variables $X,Y$, we have $\|\E X - \E Y\| = \E \|X - Y\| \pm (\E \| X - \E X\| + \E \|Y - \E Y\|)$.
Applying this fact with $(X,Y)$ being the joint distribution of $(x_i,x_j)$ conditioned on $z_S = x_S$, we can bound the contribution of the pair $(i,j)$ to the rounded cost:
\begin{align*}
    \E_{z_S \sim \mu_S} \left ( 1- \frac{\|\pE [x_i \, | \, x_S = z_S ] - \pE [x_j \, | \, x_S = z_S ]\|}{d_{i,j}} \right )^2 
    & \leq 2 \E_{z_S \sim \mu_S} \left ( 1- \frac{\pE [\| x_i - x_j \|  \, | \, x_S = z_S ]}{d_{i,j}} \right )^2 + 4 \eps \\
    & \leq 2 \E_{z_S \sim \mu_S} \pE \left [ \left ( 1- \frac{\| x_i - x_j \| }{d_{i,j}} \right )^2 \, \big | \, x_S = z_S \right ] + 4 \eps \quad \text{(Jensen's)} \, .
\end{align*}
By the law of total expectation, this is (twice) the contribution of the pair $(i,j)$ to the objective value of the LP, plus a negligible $4 \eps$.
Since the objective value of the LP is in total at most $\OPT$, if we were lucky like this for every pair $(i,j)$ we would obtain the bound $O(\OPT + \eps)$ on the rounded cost. 

As an aside, the ``global correlation rounding'' analysis of \cite{BRS,raghavendra2012approximating} can easily be adapted to show that the above variance bound holds for most $(i,j)$ pairs if we take $|S| = (\Delta / \eps)^{O(1)}$.
This would recover the result of \cite{demaine2021multidimensional}.

We aim to take $|S| \approx  (\log \Delta)/ \eps^{O(k)}$, in which case we will show that this variance bound holds for all but a roughly $(\eps + |S| \OPT)$-fraction of pairs.
To this end, we introduce a ``geometry-aware'' strategy to obtain a (weakened version of) \eqref{eq:tech-1.5}.
(We will eventually be able to argue that the small fraction of pairs where this variance bound fails do not contribute much to the rounded cost.)

\paragraph{Variances Are as Small as the Distance to the Closest Conditioned Point.}
This brings us to the question: for a given $(i,j)$ pair, how big are the variances $\pE \|x_i - \pE x_i\|^2, \pE \|x_j - \pE x_j\|^2$ after we condition on $x_S = z_S$?
Consider any jointly-distributed $\R^k$-valued random variables $(X,Y)$ -- now think of $X$ as one of the $x_i$'s, and $Y$ as some $x_\ell$ for $\ell \in S$.
Simple probability arguments show that when we fix a value $y$ for $Y$ and condition on $Y = y$, the typical conditional variance of $X$ is bounded by the typical square distance from $X$ to $Y$.
That is, conditioning on $Y$ ``localizes'' $X$ to a ball around $Y$ of squared radius $\E \|X - Y\|^2$.
Formally, 
\[
\E_{y \sim Y} ( \E[ \|X - \E[X \, | \, Y = y] \| \, | \, Y = y] ) \leq \E \|X - Y\|^2 \, .
\]

For a typical pair $(i,j)$, we want to localize $x_i$ and $x_j$ each to a ball of squared radius $\eps \cdot d_{i,j}^2$, as in \eqref{eq:tech-1.5}.
This localization happens if $S$ contains $i',j'$ such that $\pE \|x_i - x_{i'}\|^2 \leq \eps \cdot d_{i,j}^2$ and similarly $\pE \|x_j - x_{j'}\|^2 \leq \eps \cdot d_{i,j}^2$.
(Conditioning on the remaining $S \setminus \{i',j'\}$ can only reduce the expected variances further.)
To see why we should hope that $S$ will contain such $i',j'$, we need to take a detour to obtain some control over $\pE \|x_i - x_{i'}\|^2$ for $i' \in S$.

\usetikzlibrary{calc}

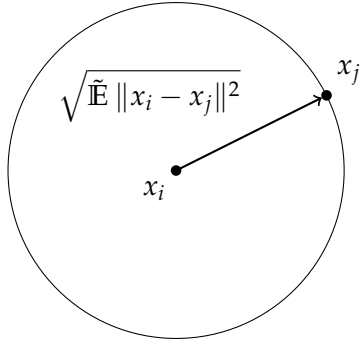
\begin{figure}[htbp]
\centering
\begin{tikzpicture}
    % Define coordinates for xi and xj
    \coordinate[label=below left:$x_i$] (xi) at (0,0);
    \coordinate[label=above right:$x_j$] (xj) at (2,1);
    
    % Draw circle centered at xi passing through xj
    \draw (xi) let \p1 = ($(xj)-(xi)$) in circle ({veclen(\x1,\y1)});
    
    % Draw points xi and xj
    \fill (xi) circle (2pt);
    \fill (xj) circle (2pt);
    
    % Draw radius and label it
    \draw[->, thick, shorten >=2pt] (xi) -- (xj) node[midway, above left] {$\sqrt{\pE \| x_i - x_j \|^2}$};
\end{tikzpicture}
\caption{Variance of $x_i$ after conditioning on $x_j$ is at most $\pE \|x_i - x_j\|^2$.}
\label{fig:circle_xi_xj}
\end{figure}

\paragraph{Most Pairwise Distances are Close to Euclidean.}
We will argue that for most $(a,b)$ pairs, $\pE \|x_a - x_b\|^2$ is close to the distance $\|x_a^* - x_b^*\|^2$, where $x_a^*$ and $x_b^*$ are members of some actual $n$-point $k$-dimensional Euclidean metric.

First, we argue that for most $(a,b)$ pairs, $\pE \|x_a - x_b\|^2 \approx d_{a,b}^2$.
Indeed, a simple Markov inequality applied to the objective function shows this.
We have $\E_{a,b} \pE ( 1- \|x_a - x_b\|/ d_{a,b})^2 \leq \OPT$.
Hence, with probability at least $1 - O(\OPT)$ over a random pair $a,b \in [n]$, we have $\pE (1 - \|x_a - x_b\| / d_{a,b})^2 \leq 0.01$.
For any such $a,b$ pair, $\pE \|x_a - x_b\|^2 = (1 \pm 0.1) \cdot d_{a,b}^2$.

We have related $\pE \|x_a - x_b\|^2$ to the input metric distances $d_{a,b}$.
Now we want to relate the input metric distances $d_{a,b}$ to distances in an actual low-dimensional Euclidean metric.
Let $x_1^*,\ldots,x_n^* \in \R^k$ be an optimal KK solution, achieving objective cost $\OPT$.
By the same Markov argument as above, with probability $1 - O(\OPT)$ over $a,b$, we have $\|x_a^* - x_b^*\|^2 = (1 \pm 0.1) \cdot d_{a,b}^2$.
Putting these together, again with probability $1 - O(\OPT)$ over $a,b$, we have $\pE \|x_a - x_b\|^2 = (1 \pm 0.2) \|x_a^* - x_b^*\|^2$.
In conclusion, if $\OPT$ is small, then for most pairs $a,b$, the pseudoexpected (square) distances agree with those of some actual $n$-point $k$-dimensional Euclidean metric, up to a multiplicative constant.

\paragraph{$x^*$-Distance to Closest Conditioned Point Is Small.}
Let us return to the mission of showing that for a typical pair $i,j$, the random set $S$ contains $i',j'$ such that $\pE \|x_i - x_{i'}\|^2 \leq \eps \cdot d_{i,j}^2$ and similarly for $j,j'$.
Since we have showed that most pairwise distances satisfy $\pE \|x_a - x_b\|^2 \approx d_{a,b}^2 \approx \|x_a^* - x_b^*\|^2$ for some $k$-dimsional metric $x_1^*,\ldots,x_n^*$, it will help to first show an analogous statement where $\pE \|x_i - x_{i'}\|^2$ and $d_{i,j}^2$ are replaced with $\|x_i^* - x_{i'}^*\|^2$ and $\|x_i^* - x_j^*\|^2$, respectively.

Consider a simple example.
Suppose $k = 1$ and $x_1^*,\ldots,x_n^*$ are the uniformly-spaced metric where $x_i^* = i$.
\begin{figure}[htbp]
\centering
\begin{tikzpicture}
    % Draw the number line
    \draw[->, thick] (-1,0) -- (7,0) node[anchor=north] {$\mathbb{R}$};
    
    % Draw tick marks and labels
    \foreach \x/\xtext in {0/{0}, 1/{$x_1^*$}, 2/{$x_2^*$}, 6/{$x_n^*$}}{
        \draw[thick] (\x,0.1) -- (\x,-0.1) node[below] {\xtext};
    }
    
    % Add ellipsis below the line to indicate continuation
    \node at (4,-0.3) {$\cdots$};
\end{tikzpicture}
\caption{Example -- uniformly-spaced metric $x_1*,\ldots,x_n^*$}
\label{fig:number_line}
\end{figure}
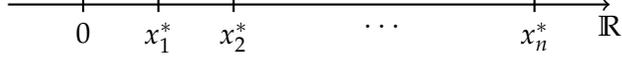
In this example, we can see that a random pair $i,j$ will have distance $\|x_i^* - x_j^*\| = \Omega(n)$.
Fixing $i$ and taking $S$ to contain $t$ random indices from $[n]$, we also have
\[
\E_S \min_{i' \in S} \|x_i^* - x_{i'}^*\| \approx \frac n {t} \approx \frac{\|x_i^* - x_j^*\|}{t}\, ,
\]
since around a $1/t$-fraction of the points $x_{i'}^*$ are at distance $\approx n/t$ to $x_i$.
If $k > 1$ but we maintain the ``uniform spacing'' assumption, we would similarly obtain $\E_S \min_{i' \in S} \|x_i^* - x_{i'}^*\| \approx \frac {n^{1/k}} {t^{1/k}} \approx \frac{\|x_i^* - x_j^*\|}{t^{1/k}}$.
So, if we take $t = 1/\eps^{O(k)}$, then at least for a typical pair $(i,j)$, we expect to have $\min_{i' \in S} \|x_i^* - x_{i'}^*\|^2 \leq \eps \cdot d_{i,j}^2$.

While the reasoning above seems specific to the uniformly-spaced metric, we can actually prove the following version of it, which applies to a general $k$-dimensional discrete metric, and also quantifies the fraction of $i,j$ pairs for which the inequality holds:
\begin{lemma}[Key lemma on $k$-dimensional discrete metrics]
\label{lem:tech-key}
Let $x_1^*,\ldots,x_n^* \in \R^k$ be a discrete metric with aspect ratio $\Delta$.
For every $\delta > 0$, with probability at least $1 - O(\delta \log(\Delta))$ over $i,j \sim [n]$ chosen uniformly at random, with probability at least $1 - \delta$ over $S \subseteq [n], |S| = t$ is uniformly random,
\[
\min_{a \in S} \|x_i^* - x_a^*\|^2 + \min_{a \in S} \|x_j^* - x_a^*\|^2 \leq \tilde{O} \left ( \frac{1}{\delta t} \right )^{1/k} \cdot \|x_i^* - x_j^*\|^2 \, .
\]
\end{lemma}

We prove Lemma~\ref{lem:tech-key} by splitting all the pairs $(i,j)$ according to $\log \Delta$ geometrically-increasing scales, and arguing separately at each scale.
The resulting dependence on $\log \Delta$ in Lemma~\ref{lem:tech-key} is the source of the dependence of our approximation guarantees on $\log \Delta$.
This dependence is tight in the sense that it cannot be removed from the statement of Lemma~\ref{lem:tech-key}, as we show in Example~\ref{ex:tight-line-metric}.
In the rest of this overview, we will ignore the failure probability over $S$ in Lemma~\ref{lem:tech-key}.
\paragraph{From $x^*$ Distances Back to Variances.}
Our goal, as before, is to show that for most $i,j$, $\E_S \min_{a \in S} \pE \|x_i - x_a\|^2 \leq \eps \cdot d_{i,j}^2$ and similarly for $j$.
(We still have to do something about the rounded cost of the small set of $(i,j)$ pairs for which this fails.)
Our progress so far is:
there is $k$-dimensional Euclidean metric $x_1^*,\ldots,x_n^*$ such that for random $i,j$,
\begin{enumerate}
    \item with probability $1 - O(\OPT)$, $\pE \|x_i - x_j\|^2 = (1 \pm 0.1) \cdot \|x_i^* - x_j^*\|^2 = (1 \pm 0.1) \cdot d_{i,j}^2$, and \label{itm:tech-1}
    \item with probability $1 - \eps$, $\min_{a \in S} \|x_i^* - x_a^*\|^2 + \min_{a \in S} \|x_j^* - x_a^*\|^2 \leq \eps \cdot \|x_i^* - x_j^*\|^2$, if we take $t \approx \tfrac{\log \Delta}{\eps^{O(k)}}$. \label{itm:tech-2}
\end{enumerate}

We would like to use \ref{itm:tech-1} to replace $\|x_i^* - x_a^*\|^2, \|x_j^* - x_a^*\|^2, \|x_i^* - x_j^*\|^2$ in \ref{itm:tech-2} with $\pE \|x_i - x_a\|^2, \pE \|x_j - x_a\|^2, \pE \|x_i - x_j\|^2$, respectively, perhaps at the cost of weakening the inequality by a multiplicative constant.
We can do this via a simple counting argument.
Because $S$ is random, every $i$ for which $\min_{a \in S} \|x_i^* - x_a^*\|^2 \gg \min_{a \in S} \pE \|x_i - x_a\|^2$ witnesses around $n/t$ pairs $i,a$ which fail the inequality in \eqref{itm:tech-1}.
So, with probability roughly $1 - O(t \cdot \OPT)$ over a randomly chosen $i$, we will have
\[
\min_{a \in S} \pE \|x_i - x_a\|^2 \leq O(1) \cdot \min_{a \in S} \|x_i^* - x_a^*\|^2  \, .
\]

Putting things together, if we draw a random pair $i,j$, with probability at least $1 - O(t \cdot \OPT + \eps)$, we will have
\begin{align}
\label{eq:tech-3}
\E_{z_S \sim \mu_S} \pE [ \| x_i - \pE[ x_i \, | \, x_S = z_S ] \|^2 \, \big | \, x_S = z_S ] + \E_{z_S \sim \mu_S} \pE [ \| x_j - \pE[ x_j \, | \, x_S = z_S ] \|^2 \, \big | \, x_S = z_S ]  \leq \eps \cdot d_{i,j}^2 \, ,
\end{align}
giving us the inequality \eqref{eq:tech-1.5} that we have been looking for.
The final step is to address the rounded cost of the ``bad'' $i,j$ pairs, occurring with probability at most $O(t \cdot \OPT + \eps)$.

\paragraph{Small Sets of Pairs Don't Contribute Much to Rounded Cost.}
Suppose that $B \subseteq [n] \times [n]$.
We show that the contribution of $B$ to the rounded objective cost is (in expectation) $O(\OPT + \Pr_{(i,j) \sim \phi}((i,j) \in B))$, giving us a bound on the contribution to the rounded cost from bad pairs.

Fix a pair $(i,j)$, a subset $S \subseteq [n]$, $|S| = t-2$, and an assignment $z_S$ for $x_S$.
The contribution of $(i,j)$ to the rounded cost obeys the bound
\begin{align}
\left ( 1- \frac{\|\pE [x_i \, | \, x_S = z_S ] - \pE [x_j \, | \, x_S = z_S ]\|}{d_{i,j}} \right )^2  & \leq 1 + \frac{\|\pE [x_i \, | \, x_S = z_S ] - \pE [x_j \, | \, x_S = z_S ]\|^2}{d_{i,j}^2} \nonumber \\
& \leq 1 + \left ( \pE \left [ \frac{\|x_i - x_j \|}{d_{i,j}} \, | \, z_S = x_S \right ]^2 \right )  \quad \text{(Jensen's)} \nonumber \\
& \leq 3 + 2 \cdot \pE \left [ \left ( 1 - \frac{\|x_i - x_j\|}{d_{i,j}} \right )^2  \, | \, z_S = x_S \right ]  \, , \label{eq:tech-1}
\end{align}
where for the last step we added and subtracted $1$ from the expression inside the square.
If we take the expectation of \eqref{eq:tech-1} over $z_S \sim \mu_S$,
we get $O(1) + (\text{contribution of $i,j$ to LP objective})$.
So,
\begin{align}
\E_{i,j} \E_{z_S \sim \mu_S} \left ( 1- \frac{\|\pE [x_i \, | \, x_S = z_S ] - \pE [x_j \, | \, x_S = z_S ]\|}{d_{i,j}} \right )^2 \cdot \mathbf{1}_{(i,j) \in B} \leq O(|B|/n^2 + \OPT) \, . \label{eq:tech-2}
\end{align}

\paragraph{Putting It All Together.}
Putting together our conclusions \eqref{eq:tech-3} about variance (after conditioning) of all but an $O(t \cdot \OPT + \eps)$-fraction of pairs \eqref{eq:tech-2} about the contribution of those remaining pairs,
and recalling that we chose $t \approx  (\log \Delta)/\eps^{O(k)}$, we obtain the bound
\begin{align}
\label{eq:tech-4}
(\text{rounded cost}) \leq   \frac{\log \Delta}{\eps^{O(k)}}  \cdot \OPT + \eps \, .
\end{align}
The minimum overall $\eps$ which can be achieved above is $\log \Delta \cdot \OPT^{\Omega(1/k)}$, leading to our final bound.
As a reminder, in this overview, we have ignored some technical details -- see Sections~\ref{sec:algo-and-analysis} and beyond for rigorous proofs.

\section{Open Problems}
\label{sec:open-problems}
Our work gives, as far as we know, the first polynomial-time approximation algorithm for any MDS-style objective function with provable approximation guarantees when $\Delta$ is super-logarithmic in $n$.
We hope that the idea of using convex programming hierarchies to Kamada-Kawai opens the door to further progress on related optimization-based formulations of dimension reduction, MDS, and metric embedding.
To this end, we offer a non-exhaustive list of some open problems, focusing on MDS-style objectives.
% \Snote{there are also lots of open problems in algorithms for metric embeddings with other notions of distortion but that seems to be getting very broad}

\paragraph{Improved algorithms for Kamada-Kawai.}
The only hardness of approximation we are aware of for the Kamada-Kawai objective is the result of \cite{demaine2021multidimensional} showing that an FPTAS is not possible unless $P = NP$.
But this does not rule out, say, constant-factor approximation:

\begin{problem}
    Design a constant-factor approximation algorithm for the Kamada-Kawai function which is polynomial time for superlogarithmic aspect ratio.
\end{problem}

\paragraph{Weighted least squares.}
A wider range of MDS-style objectives (many available in R, scikit-learn, Matlab, etc.) is captured via the following generalization of Kamada-Kawai.
A weighted least-squares instance is specified by dissimilarities $\{d_{ij}\}_{i,j \in [n]}$ and nonnegative weights $\{w_{ij}\}_{i,j \in [n]}$:
\[
\min_{x_1,\ldots,x_n \in \R^k} \E_{i \neq j} w_{ij} \Paren{1 - \frac{\|x_i - x_j\|}{d_{ij}}}^2 = \min_{x_1,\ldots,x_n \in \R^k} \E_{i \neq j} \frac{w_{ij}}{d_{ij}^2} ( d_{ij} - \|x_i - x_j\|)^2 \, .
\]
Setting all $w_{ij} = 1$ recovers the Kamada-Kawai objective; other choices lead to other popular MDS objectives.
For instance, $w_{ij} = d_{ij}$ is ``Sammon mapping'' and $w_{ij} = d_{ij}^2$ is ``raw stress''.
By analogy to CSPs, it's natural to associate a parameter to the weights measuring their density: let $W = \max w_{ij} / \E w_{ij}$.
The algorithm of \cite{demaine2021multidimensional} extends to an additive PTAS for weighted MDS, with running time exponential in $W$.
The case of ``raw stress'' is particularly interesting, because it corresponds to least-squares projection of the given metric into the set of $k$-dimensional Euclidean distance matrices.

\begin{problem}
    Provide nontrivial polynomial-time approximation guarantees for weighted least-squares when $\Delta$ is superlogarithmic.
\end{problem}

\paragraph{Krustal Stress-1 and Stress-2.}
Another common objective function is Kruskal's ``Stress-$1$'':
\[
\text{Stress-$1$}(x_1,\ldots,x_n) = \frac{ \E_{i \neq j} (d_{ij} - \|x_i - x_j\|)^2}{\E_{i \neq j} \|x_i - x_j\|^2} \, .
\]
Stress-$1$ is qualitatively different from the objectives we discussed already because of the normalization by $\|x_i - x_j\|$, meaning that the objective value penalizes contracting pairs somewhat more than the weighted least-squares objective.
A local-search routine for minimizing Stress-1 is the default implementation in Matlab's built-in \texttt{mdscale}, and is also available in scikit-learn.
Stress-$2$ is yet another variant:
\[
\text{Stress-$2$}(x_1,\ldots,x_n) = \frac{ \Paren{ \E_{i \neq j} (d_{ij} - \|x_i - x_j\|)^2} ^2}{\E_{i \neq j} \|x_i - x_j\|^4} \, .
\]

\begin{problem}
    Provide a nontrivial polynomial-time approximation algorithm for Kruskal Stress-$1$ or Stress-$2$.
\end{problem}

\section{Preliminaries}
\label{sec:prelims}
% We use the following notation
% \begin{enumerate}
%     \item $X^* = \Set{x_i^*}_{i \in [n]}$ is an embedding that achieves optimal cost
%     \item $\Set{d_{i,j}}_{i,j \in [n]} $ is a set of input distances, \raj{scaled so that $\min_{i,j} d_{i,j} = 1$} \ainesh{we don't need this assumption}
%     \item $k$ denotes the dimension of the target space.
%     \item We work with the normalized squared objective:
%     \begin{equation*}
%         \min_{x_1, \ldots , x_n \in \mathbb{R}^k } \frac{1}{n^2} \sum_{i, j \in [n]} \Paren{ 1- \frac{\norm{ x_i - x_j}_2 }{d_{i,j}} }^2 
%     \end{equation*}
% \end{enumerate}

\paragraph{Notation.}
Throughout, we will write $\expecf{i,j}{\cdot}$ to denote $\expecf{i,j \sim \binom{n}{2}}{\cdot}$, and take a similar convention for $\max_{i,j},\min_{i,j}$ and so on. 
$\|\cdot \|$ denotes Euclidean norm on $\R^k$.
We write $\calN_{k, \Delta,\eps}$ to denote an $\eps$-cover / $\eps$-net of $[-\Delta ,\Delta]^{\otimes k}$, i.e. $[-\Delta ,\Delta]^{\otimes k} \supseteq \calN_{\Delta, k,\eps}$ and for any $v \in [-\Delta ,\Delta]^{\otimes k}$, there exists a $v' \in  \calN_{\Delta, k,\eps}$ such that $\norm{v - v'}\leq \eps$. It is standard, for instance via greedily adding uncovered points to $\calN_{k, \Delta, \eps}$, that such covers exist with
$\Abs{  \calN_{k, \Delta,\eps}  } \leq \Paren{c \Delta/\eps}^{k}$, for a fixed constant $c$, and that moreover one can have $\min_{x,y \in \calN_{k, \Delta,\eps}} \|x-y\| > \eps$. For a finite set $S \subseteq \R$, we write $\calQ(S,\eps)$ to denote the value of the $\eps$-th quantile/$\eps|S|$-th least element in $S$.

\vspace{-0.15in}
\paragraph{Sherali-Adams Linear Programming.}
We review the local-distribution view of the Sherali-Adams linear programming hierarchy.
Let $\Omega$ be a finite set and $n \in \N$.
For $t \in \N$, the level-$t$ Sherali-Adams LP relaxation of $\Omega^n$ is as follows.
(For a more formal account, see e.g. \cite{BRS}.)

\vspace{-0.15in}
\paragraph{Variables and constraints.} For each $T \in \binom{n}{t}$, we have $|\Omega|^t$ non-negative variables, $\mu_T(x)$ for $x \in \Omega^t$, collectively denoted $\mu_T$.
For each $T$, we constrain $\sum_x \mu_T(x) = 1$ so that $\mu_T$ is a probability distribution over $\Omega^t$.
Then, for every $S,T \in \binom{n}{t}$, we add linear constraints so that the distribution $\mu_T$ restricted to $S \cap T$ is identical to the distribution $\mu_S$ restricted to $S \cap T$.
Overall, we have $(|\Omega|n)^{O(t)}$ variables and constraints.
We sometimes call the collection of local distributions $\{ \mu_T \}$ satisfying these constraints a ``degree-$t$ Sherali-Adams pseudo-distribution''.

\vspace{-0.15in}
\paragraph{Pseudo-expectation view.}
If we have a Sherali-Adams solution which is clear from context, for any function $f$ of at most $t$ variables out of $x_1,\ldots,x_n$ (a.k.a. ``$t$-junta''), we write $\pE f(x_1,\ldots,x_n)$ for $ \E_{x \sim \mu_T} f(x_T)$, where $\mu_T$ is the local distribution on the subset $T \subseteq [n]$ of variables on which $f$ depends.
We can extend $\pE$ to a linear operator on the linear span of $t$-juntas.
Observe that $\pE f$ is a \emph{linear} function of the underlying LP variables.

\vspace{-0.15in}
\paragraph{Conditioning.}
We will frequently pass from a Sherali-Adams solution $\{\mu_T\}$ to a \emph{conditional} solution $\{\mu_T \, | \, x_{i_1} = a_1,\ldots,x_{i_{t'}} = a_{t'} \}$, for $t' < t$.
We obtain the conditional local distribution on $x_{j_1},\ldots,x_{j_{t-t'}}$ by conditioning the local distribution on $x_{j_1},\ldots,x_{j_{t-t'}}, x_{i_1},\ldots,x_{i_{t'}}$ on the event $\{ x_{i_1} = a_1,\ldots,x_{i_{t'}} = a_{t'} \}$.

\section{Algorithm and Analysis}
\label{sec:algo-and-analysis}

In this section, we formally state and prove our main theorem for the \textit{Kamada-Kawai} objective.

\begin{theorem}[Main Theorem]
\label{thm:efficient-algo-mds}
Given an instance of MDS under the 
 Kamada-Kawai objective, with  aspect ratio $1\leq \Delta$, target dimension $k \in \mathbb{N}$, and target accuracy $0<\eps < 1$, there exists an algorithm that outputs an embedding $\Set{\hat{x}_i }_{i\in[n]}$ such that with probability at least $99/100$,
\begin{equation*}
    \expecf{i,j}{\Paren{ 1 - \frac{\norm{\hat{x}_i - \hat{x}_j }_2 }{d_{i,j}} }^2} \leq 
    \begin{cases}
        \vspace{0.05in} \bigO{ \sqrt{\OPT} \cdot \log(\Delta/\eps) } + \eps  & \textrm{if } k= 1\\
        \vspace{0.05in} \bigO{ \sqrt{\OPT}\cdot \log(1/\eps) \cdot \log(\Delta/\eps)  } + \eps & \textrm{if } k= 2 \\
        \bigO{ \OPT^{ \hspace{0.04in}1/k } \cdot \log(\Delta/\eps)  } + \eps & \textrm{otherwise.}
    \end{cases}
\end{equation*}
Further, the algorithm runs in $\Paren{n\Delta/\eps}^{\mathcal{O}\Paren{k \cdot \tau }}$ time, where 
\begin{equation*}
\tau = 
    \begin{cases}
    \vspace{0.05in} \bigO{  \Paren{ \log\Paren{\log(\Delta)/\eps} \log\Paren{\Delta/\eps} } /\eps   }  & \textrm{ if } k = 1 \\
    \vspace{0.05in} \bigO{  ( \log^2\Paren{\log(\Delta)/\eps} \log\Paren{\Delta/\eps} ) /\eps   }  & \textrm{ if } k = 2 \\
    \bigO{ \Paren{ k \log(\log(\Delta)/\eps)\log(\Delta/\eps)^{k/2} } /\eps^{k/2} } & \textrm{ otherwise. } 
    \end{cases}    
\end{equation*}

\end{theorem}

\begin{remark}
\label{rem:k-dependence}
See Section~\ref{sec:smaller-LP} for a slight strengthening of Theorem~\ref{thm:efficient-algo-mds} where we sparsify the LP to obtain a $\poly(n)$ running time for constant $\Delta,k,\eps$.

We have optimized the choices of parameters in Theorem~\ref{thm:efficient-algo-mds} to maximize the exponent of $\OPT$ in our approximation guarantee.
Different tradeoffs among the parameters are possible.
As one example, if we follow an argument in line with technical overview, we can obtain a running time with $\Delta$-dependence of the form $\exp(( \log \Delta)^{O(1)} / \eps^{O(k)})$ rather than $\exp(((\log \Delta) / \eps))^{O(k)})$, at the cost that instead of $\OPT^{1/k}$ we get $\OPT^{\Omega(1/k)}$.
\end{remark}

\begin{mdframed}
  \begin{algorithm}[Efficient Kamada-Kawai]
    \label{algo:efficient-algo}\mbox{}
    \begin{description}
    \item[Input:] Nonnegative numbers $\calD =\Set{ d_{i,j} }_{i,j \in [n]}$, target dimension $k \in \mathbb{N}$, target accuracy $0<\eps<1$.
    
    \item[Operation:]\mbox{}
    
\begin{compactenum}
    \item Let \begin{equation*}
    \tau = 
    \begin{cases}
    \bigO{  \Paren{ \log\Paren{\log(\Delta)/\eps} \log\Paren{\Delta/\eps} } /\eps   }  & \textrm{ if } k = 1 \\
    \bigO{  ( \log^2\Paren{\log(\Delta)/\eps} \log\Paren{\Delta/\eps} ) /\eps   }  & \textrm{ if } k = 2 \\
    \bigO{ \Paren{ k \log(\log(\Delta)/\eps)\log(\Delta/\eps)^{k/2} } /\eps^{k/2} } & \textrm{ otherwise } 
    \end{cases}
    \end{equation*}
    Let $\mathcal{N}_{\Delta/\eps, k, 1}$ be an $1$-net of $[-\Delta/\eps,\Delta/\eps]^k$.
    Let $\mu$ be a degree-$(2\tau)$ Sherali-Adams pseudo-distribution over the domain $\mathcal{N}_{k, \Delta/\eps, 1}$ optimizing
    \[
    \min_\mu \, \pE_\mu \E_{i,j \sim [n]} \Paren{ 1 - \frac{\|x_i - x_j\|}{d_{ij}}}^2 \, .
    \]
    \item Let $\calT\subset [n]$ be a random subset of size $\tau$ and let $\tilde{x}_\calT $ be a draw from the local distribution $\{x_\calT\}$. Let $\mu_\calT$ be the pseudo-distribution obtained by conditioning $\mu$ on $\zeta_\calT := \Set{ x_j = \tilde{x}_j }_{j \in \calT}$.   
%    \item Let $\mu_\calT$ be the pseudo-distribution obtained by conditioning on the event $\zeta$.      
\end{compactenum} 
    \item[Output:]  The embedding $\Set{ \pexpecf{\mu_\calT}{ x_i }  }_{i \in [n]}$. % satisfying the guarantee of Theorem~\ref{thm:efficient-algo-mds}.
    \end{description}
  \end{algorithm}
\end{mdframed}

Our goal is to set up a Sherali-Adams LP relaxation for the Kamada-Kawai objective.
The version of Sherali-Adams we use here works for discrete assignment problems, while Kamada-Kawai asks us to assign $x_1,\ldots,x_n$ to vectors in $\R^k$, so we need to discretize $\R^k$.
The following lemma shows that only an additive $\eps$ is lost in the objective value if we discretize to $\bigO{\Delta / \eps}^{k}$ points.
We defer the proof to Section~\ref{sec:discretization}.
\begin{lemma}[Aspect ratio of the target space]
\label{lem:aspect-ratio}
Let $\Set{d_{i,j}}_{i,j \in [n]}$ satisfy $1 \leq d_{i,j} \leq \Delta$ for all $(i,j) \in \binom{n}{2}$ and let $\OPT$ be the optimal Kamada-Kawai objective value with respect to $\R^k$. 
For $0<\eps<1$, let $\mathcal{N}_{k, \Delta/\eps, 1}$ be a $1$-cover of $[-\Delta/\eps, \Delta/\eps]^k$.
There exists an embedding $x_1,\ldots,x_n \in \mathcal{N}_{k,  \Delta/\eps, 1}$, such that for all $(i,j) \in \binom{n}{2}$ we have:
\begin{equation*}
    1 \leq \|x_i-x_j\|_2 \leq \bigO{\Delta/\eps}
    \quad \text{ and } \quad \E_{i,j \in [n]} \Paren{ 1- \frac{\norm{x_i - x_j}  }{d_{i,j} } }^2 \leq \OPT + \eps \, .
\end{equation*}
%\ainesh{we need to embed into a cover instead of a net, small edit req here.}
%\Raj{A $\eps$-net is an $\eps$-cover?}
\end{lemma}

We are now ready to describe our algorithm.
We follow a fairly standard strategy for rounding Sherali-Adams relaxations: condition on the values of a few randomly-chosen variables and output a simple function of the $1$-wise local distributions on the remaining variables (in this case, simply their expectations).
Our main contribution is a new \emph{geometry-aware} analysis of this rounding procedure, which obtains provable guarantees with $\approx \log(\Delta)/\eps$ levels of the Sherali-Adams hierarchy when embedding into the line or the plane.

In what follows, we use \emph{pseudo-distribution} to refer to any $n$-variable Sherali-Adams pseudo-distribution over the domain $\calN_{k, \Delta/\eps, 1}$.

\subsection{Analysis setup and pseudo-deviations}

% \ainesh{add hard example with large deviations}
Our first key innovation is to track the following measure of the fluctuations/uncertainty in the local distributions in our (conditioned) pseudo-distribution.

\begin{definition}[Pseudo-Deviation of an Embedding]
Given a pseudo-distribution $\mu$, we define the pseudo-deviation of each point as follows:
\begin{equation*}
   \pdev{\mu}{x_i } = \pexpecf{\mu}{ \Norm{  x_i -  \pexpecf{\mu}{x_i} }_2 } .
\end{equation*}
\end{definition}
At a high level, we argue that 
\begin{enumerate}[(a)]
\item if  the conditional pseudo-deviations are small, i.e. $\pdev{\mu_\calT}{x_i}, \pdev{\mu_\calT}{x_j} \leq \delta \cdot d_{ij}$, our rounding algorithm incurs cost that is a function of $\delta$, and \label{taga}
\item for most pairs $i,j$, the conditioned pseudo-distribution $\mu_\calT$ satisfies $\pdev{\mu_\calT}{x_i}, \pdev{\mu_\calT}{x_j} \ll d_{ij}$.
\label{tagb}
\end{enumerate}

\paragraph{On the need for conditioning.}
We show that the latter holds, despite the fact that for $\mu$ (as opposed to $\mu_\calT$), we could have $\pdev{\mu}{x_i} \approx \Delta$ for most $i$. To see this, consider a one-dimensional instance where the optimal embedding is to map half the points to roughly $-\Delta$ and half the points to $\Delta$. In this example, the target embedding, $\Set{x_i}_{i \in [n]}$,  can be partitioned into two clusters $\calC_1, \calC_2$, such that one gets mapped to $-\Delta$ and another gets mapped to $\Delta$. Now, assume the pseudo-distribution $\mu$ is the uniform distribution over $(\calC_1 \to -\Delta, \calC_2 \to \Delta)$ and $(\calC_1 \to \Delta, \calC_2 \to -\Delta)$. Such a pseudo-distribution is valid since the cost remains unchanged. However, for each $i \in[n]$, $\pdev{\mu}{x_i} = 2\Delta$ since $\pexpecf{\mu}{x_i} =0$ and $x_i$ is $\pm \Delta$ with equal probability. However, observe that conditioning on the location of a single $x_i$ precludes such a distribution over solutions. With this intuition in mind, we show that conditioning in step (2) of Algorithm~\ref{algo:efficient-algo} reduces $\pdev{}{x_i}$ in general.

\paragraph{Rounded cost to pseudo-deviations.}
We capture (\ref{taga}) via the following simple fact, which we will apply with $X,Y$ being a joint sample from the conditional local distribution on $x_i,x_j$. Since we restrict attention to pairs of variables, we can freely translate between distributions and pseudo-distributions. 

\begin{fact}[Pseudo-expected distances to pseudo-deviations, see proof in Section~\ref{sec:deferred-proofs-var-reduction-discretization}]
\label{fact:rounding-deviations}
    Let $X,Y$ be vector-valued random variables and let $d > 0$.
    Then
    \[
    \Paren{ 1 - \frac{ \| \pE X - \pE Y \| } {d}}^2 \leq 2 \pE \Paren{ 1 - \frac{ \| X - Y \| } {d}}^2 + 2 \Paren{ \frac{\pE \|X - \pE X\| + \pE \|Y - \pE Y\|}{d}}^2 \, .
    \]
\end{fact}
In our setting, the left-hand side of \cref{fact:rounding-deviations} represents the cost of $x_i,x_j$ in the rounded solution, while the right-hand side is the sum of two terms: the contribution of the same pair to the Sherali-Adams objective function and a term depending on the pseudo-deviations of $x_i$ and $x_j$, scaled by the distance in the input space, i.e. $(\pdev{\mu_\calT}{x_i} + \pdev{\mu_\calT}{x_j})/d_{ij}$.

\paragraph{Pseudo-deviations to quantiles.}
To capture (\ref{tagb}), we prove the following lemma, which says that $\pdev{\mu_\calT}{x_i}^2$ can be bounded in terms of the $\eps$-th quantile of the set of numbers $\Set{ \pexpecf{\mu}{\|x_i - x_j\|^2 } }_{j \in [n]}$, which denote the pseudo-expected distances, as long as the number of points we conditioned on is large enough. 

\begin{lemma}[Pseudodeviation Reduction via Quantiles]
\label{lem:var-redux-quantiles}
Given $0<\eps, \delta < 0$, $1\leq \Delta$, let $\mu$ be a degree-$\bigO{ \log(1/\delta) /\eps}$  pseudo-distribution. 
Let $\calT\subset [n]$ be a random subset of size $ \log(1/\delta) /\eps$ and let $\tilde{x}_\calT $ be a draw from the local distribution $\{x_\calT\}$.
Let $\mu_\calT$ be the pseudo-distribution obtained by conditioning on $\zeta_\calT := \Set{ x_j = \tilde{x}_j }_{j \in \calT}$, as in step (2) of Algorithm~\ref{algo:efficient-algo}. 

Then, for all $i\in [n]$, with probability at least $1-\delta$ over the choice of $\calT$,
\begin{equation*}
    \expecf{  \tilde{x}_\calT }{    \pdev{\mu_\calT}{ x_i }^2 }  \leq \calQ\Paren{ \Set{ \pexpecf{\mu}{ \norm{x_i - x_j }^2 }  }_{j \in [n]}, \eps  } 
\end{equation*}
where $\calQ\Paren{ \calS, q }$ denotes the $q$-th quantile of the ordered set $\calS$.
\end{lemma}

We defer the full proof Section~\ref{sec:varredux-proof}.
% \begin{proof}[Proof sketch -- full proof deferred to Section~\ref{sec:varredux-proof}]
%   Fix $i \in [n]$.
%   Suppose that $j \in \calT$.
%   We have $$\pexpecf{\mu_{\calT}}{\norm{ x_i - \pexpecf{\mu_\calT} {x_i} }^2} \leq \pexpecf{\mu_{\calT}}{ \norm{ x_i - \pexpecf{\mu_{\calT}}{x_j} }^2 } \leq \pexpecf{\mu_\calT}{ \norm{ x_i - \tilde{x}_j }^2 }.$$
%   For the first inequality, we used that for any random vector $X$, $\E \|X - v\|^2$ is minimized by choosing $v = \E X$.
%   By law of total expectation, $\E_{\tilde{x}_j} \pE_{\mu_\calT} \|x_i - \tilde{x}_j\|^2 = \pE_\mu \|x_i - x_j\|^2$.
%   Finally, to obtain the statement about $\calQ\Paren{ \Set{ \pexpecf{\mu}{ \norm{x_i - x_j }^2 }  }_{j \in [n]}, \eps  }$, just note that if $\calT \gg 1/\eps$ then it is likely to contain some $j$ among those with the $\eps n$ smallest values $\pE_{\mu} \|x_i - x_j\|^2$.
% \end{proof}
Next, we develop tools to relate $\calQ\Paren{ \Set{ \pexpecf{\mu}{ \norm{x_i - x_j }^2 }  }_{j \in [n]}, \eps  }$ to $\{ d_{ij} \}$.
Before turning to those tools, though, what do we do with those $i,j$ where $\pdev{\mu_\calT}{x_i}, \pdev{\mu_\calT}{x_j}$ remains large even after conditioning?

\paragraph{Rounded cost on any small subset is small.}
We show that no pair $i,j$ has a rounded cost greater than $O(1)$ plus its contribution to the relaxed objective function, so if there aren't too many ``bad'' pairs like this, their contribution to the rounded cost isn't large. Formally, we show that the rounded cost on a small set of pairs can be bounded by the size of the set and their contribution to $\OPT$:

\begin{lemma}[Rounded cost on a small set of pairs]
\label{lem:sdp-val-small-subset}
 Let $\calT\subset [n]$ be a subset of size $\tau$ and let $\tilde{x}_\calT \in \Paren{ [-\Delta,\Delta]^{\otimes k} }^{|\calT|}$ be a draw from the local distribution $\{x_\calT\}$. Let $\mu_\calT$ be the pseudo-distribution obtained in Algorithm~\ref{algo:efficient-algo}, step (2).
 Let $\calS \subseteq \binom{n}{2}$ be any subset of pairs. Then, 
    \begin{equation*}
      \expecf{\tilde{x}_\calT } { \sum_{(i,j) \in \calS} \Paren{ 1 - \frac{\norm{\pexpecf{\mu}{ x_i } -  \pexpecf{\mu}{ x_j } } }{d_{i,j}} }^2 } \leq 4 \cdot \abs{\calS} + 4 \OPT\cdot \binom{n}{2}. 
    \end{equation*}
\end{lemma}
\begin{proof}
Using the fact that $(a +b)^2 \leq 2( a^2 + b^2 )$, we have
\begin{align*}
        \expecf{ \tilde{x}_{\calT} }{ \sum_{(i,j) \in \calS} \Paren{ 1 - \frac{\norm{\pexpecf{\mu}{ x_i } -  \pexpecf{\mu}{ x_j } } }{d_{i,j}} }^2 } & \leq 2 \sum_{(i,j) \in \calS } 1 + 2  \expecf{ \tilde{x}_{\calT} }{ \sum_{(i,j)\in \calS} \Paren{ \frac{\norm{  \pexpecf{\mu_{\calT}}{x_i - x_j}  }}{d_{i,j}} }^2 } \\
        & \leq 2 \cdot \Abs{\calS } + 2  \sum_{(i,j)\in \calS} \expecf{ \tilde{x}_{\calT} }{ \pexpecf{\mu_{\calT} }{ \Paren{ \frac{\norm{x_i - x_j} }{d_{i,j}} + 1 - 1 }^2 } } \\
        & \leq 4 \cdot \abs{\calS} + 4  \sum_{(i,j) \in \binom{n}{2}} \pexpecf{\mu}{ \Paren{ \frac{\norm{x_i - x_j} }{d_{i,j}}  - 1 }^2 }\\
        & \leq  4 \cdot \abs{\calS} + 4 \OPT\cdot \binom{n}{2}
\end{align*}
where the second inequality follows from Jensen's inequality on the local distributions in $\mu_\calT$, the third uses the law of total probability, and the last inequality follows from the Sherali-Adams objective value being smaller than $\OPT$.
\end{proof}

\subsection{Bounding quantiles of pseudo-expected distances }

Lemma~\ref{lem:var-redux-quantiles} shows that we can control $\pdev{\mu_\calT}{x_i}$ using quantiles of $\{ \pE \|x_i - x_j\| \}_{j \in [n]}$; now we need to bound those quantiles.
Our goal is to show that for most $i,j$,
\[
\calQ \Paren{ \{ \pE \|x_i - x_\ell \| \}_{\ell \in [n]}, \eps }, \calQ \Paren{ \{ \pE \|x_j - x_\ell \| \}_{\ell \in [n]}, \eps } \ll d_{ij}
\]

As a thought experiment, consider the case where instead of a pseudo-distribution we had an actual set of points $x_1^*,\ldots,x_n^* \in \R^k$, 
 and instead of $d_{ij}$ we had $\|x_i^* - x_j^*\|$, we could hope that quantiles of $\calQ ( \{ \|x_i^* - x_\ell^* \| \}_{j \in [n]}, \eps ) \ll \|x_i^* - x_j^* \|$ for typical $\ell \in [n]$.
Interestingly even in our setting, we can establish a fact like this for sets of points $x_1^*,\ldots,x_n^* \in \R^k$, and then show that most $\pE \|x_i - x_j\|$ behave like a corresponding $\|x_i^* - x_j^*\|$ for an optimal embedding $x_1^*,\ldots,x_n^*$.

\paragraph{Low-dimensional geometry and quantiles of pair-wise distances.}

Consider the special case of n distinct points, $\Set{x_1^*, x_2^*, \ldots, x_n^*}$, on the discrete interval $[-\Delta, \Delta]$. We derive the following geometric fact: ignoring constant factors, for all but an $\eps\cdot \log(\Delta)$-fraction of pairs $(i,j)$, the intervals of radius $\eps\cdot \norm{x_i^* - x_j^*}$ around $x_i^*$ and $x_j^*$ respectively contain at least $\eps n$ points. This statement immediately implies that the $\eps$-th quantiles of the sets $\Set{ \norm{ x_i^* - x_k^* } }_{k \in [n]}, \{ \norm{ x_j^* - x_k^* } \}_{k \in [n]}$, are at most $\eps\cdot \norm{x_i^* - x_j^*}$. %Further, discarding $\eps\cdot \log(\Delta)$-fraction of pairs $(i,j)$ is necessary in the worst case (see \cref{sec:hard-inst-quantiles})\ainesh{add this section if time remains.}.

We state a generalized version of this geometric fact below and defer the proof to \cref{sec:quantiles-of-pairwise-distances}:
\begin{lemma}[Quantiles of distances in low-dimensional space]
\label{lem:quantiles-with-subsets}
Given $0< \eta, \delta <1$, and a set of $n$ points in the range $[-\Delta, \Delta]^{\otimes k}$, denoted by $\Set{x_i^*}_{i \in [n]}$,  for all but  $\bigO{ \delta \log(\Delta)\cdot n^2 }$ pairs, $(i,j)$ such that $\norm{ x_i^* - x_j^* } \geq 1$, we have
\begin{equation*}
 \calQ\Paren{ \Set{\norm{ x_i^* - x_k^* } }_{k \in  [n] }, \eta  }  + \calQ\Paren{ \Set{\norm{ x_j^* - x_k^*} }_{k \in [n] }, \eta }  \leq \bigO{ \Paren{ \frac{ \eta }{\delta} }^{1/k} } \norm{x_i^* - x_j^*}. 
\end{equation*}
\end{lemma}

Since the $\log(\Delta)$ in Lemma~\ref{lem:quantiles-with-subsets} is the origin of the $\log(\Delta)$ in our approximation guarantees and running times, we give the following example to show that it cannot be removed from Lemma~\ref{lem:quantiles-with-subsets}.
\begin{example}
\label{ex:tight-line-metric}
Consider the case $k=1$ where we let $n = \log \Delta$ and $x_i^* = 2^i$.
Then for every pair $i<j$, we have $2^{j-1} \leq \|x_i^* - x_j^*\| \leq 2^{j}$, thus there does not exist even a single point $x_k^*$ such that $\|x_j^* - x_k^*\| \leq (1/3) \cdot \|x_i^* - x_j^*\|$. 
\end{example}

\paragraph{Quantiles of pseudo-expected distances to quantiles of an optimal embedding.}

Next, we translate between quantiles of pseudo-expected distances and quantiles of distances in an optimal embedding.
To start, the following ``glorified Markov'' inequality shows that $\pE \|x_i - x_j\|$ is close to $\|x_i^* - x_j^*\|$ for some vectors $x_1^*,\ldots,x_n^*$ and most $i,j$ pairs. This lemma hinges on the fact that the cost of the solution output by the LP is at most $\OPT$, by virtue of it being a relaxation. It also crucially uses the fact that the Kamada-Kawai objective is an average over pairs. 

% \begin{figure}
%     \centering
%     \includegraphics[scale=1.2]{fig_shrinking_lemma.pdf}
%     \caption{The figure shows that for two random points $x_i,x_j$ on the line, we have that there are $\epsilon n$ points around $x_i$
%     and $x_j$ that shrink by $\epsilon |x_i - x_j|$. }
%     \label{fig:quantile-shrinking}
% \end{figure}

\begin{lemma}[Relating $\OPT$ to distortion on typical distances]
\label{lem:relating-opt-to-dist-on-typical}
Let $\Set{x_1^*, x_2^*, \ldots, x_n^*}$ be some optimal embedding of the points minimizing the Kamada-Kawai objective.
Then, for any  $0<c<1/4$, with probability at least $1-(\OPT/c)$, for a random pair $i,j \sim [n]$, 
\begin{equation*}
   (1 - 2\sqrt{c} ) \norm{ x^*_{i} - x^*_j }^2    \leq d_{i,j}^2 \leq (1 + 2\sqrt{c}) \norm{ x^*_{i} - x^*_j }^2.
\end{equation*}
Let $\mu$ be any pseudo-distribution that minimizes the objective $\E_{i,j\in[n] } \Paren{ 1 - \frac{\norm{x_i - x_j} }{d_{i,j} } }^2 $. Then, with probability at least $1-(\OPT/c)$, for a random pair $i,j \sim [n]$, 
\begin{equation*}
     \Paren{ 1-2\sqrt{c} }  \norm{x_i^* -x_j^* } \leq \pexpecf{\mu}{ \norm{x_i  - x_j } }  \leq  \Paren{ 1+2\sqrt{c} }  \norm{x_i^* -x_j^* },
\end{equation*}
and
\begin{equation*}
    \Paren{1-2\sqrt{c}}d_{i,j}^2   \leq \pexpecf{\mu}{ \norm{x_i  - x_j }^2 } \leq \Paren{2 + 2c} d_{i,j}^2. 
\end{equation*}
\end{lemma}

Observe, if \cref{lem:relating-opt-to-dist-on-typical} held for all the pairs, we would immediately be able to translate between quantiles of pseudo-expected distances and quantiles of distances corresponding to an optimal embedding. However, \cref{lem:relating-opt-to-dist-on-typical} holds only for a $1-\OPT$-fraction of pairs. Therefore, there could be a subset of $\sqrt{\OPT} n$ points such that for each one, their $\sqrt{\OPT}$-th quantiles do not shrink at all. Now, if either endpoint of a pair (i,j)  lands in this subset, we cannot obtain a non-trivial bound on the corresponding quantiles. There can be at most $\sqrt{\OPT} n^2$ such pairs, and we can bound the cost on this subset by invoking \cref{lem:sdp-val-small-subset} again. In the following lemma, we show that this is essentially the only thing that can go wrong and defer the proof to \cref{sec:translating}. 

%The following lemma follows quickly from the previous one:
%\ainesh{this lemma is the bottleneck between $\OPT$ and $\sqrt{\OPT}$...we should explain why.}

\begin{lemma}[Translating between quantiles]
\label{lem:translating-bet-quantiles}
Let $\mu$ be a pseudo-distribution minimizing the objective in Algorithm~\ref{algo:efficient-algo} and let $\Set{x_i^* }_{i \in [n]}$ be an optimal embedding.
For all but $(\sqrt{\OPT} +\eps)$-fraction of indices $i$, 
%with probability at least $0.99$ over the choice of $\calT$, 
\begin{equation*}
     \calQ\Paren{ \Set{ \pexpecf{\mu}{ \norm{x_i - x_j }^2 }  }_{j \in [n]}, \eps  }      \leq 10  \cdot \mathcal{Q}\Paren{ \Set{  \norm{x_i^* - x_k^* }^2  }_{k \in [n]} , 16\sqrt{\OPT} +\eps  }  .
\end{equation*}
\end{lemma}

% We defer the proofs of the above lemmas to Section~\ref{sec:distance-quantile-proofs}.

\subsection{Proof of Main Theorem}
We are now ready to complete the proof of Theorem~\ref{thm:efficient-algo-mds}.

\begin{proof}[Proof of Theorem~\ref{thm:efficient-algo-mds}]

Our goal is to bound the objective value of the rounded solution: 
\begin{equation}
\label{eqn:objective-cost}
    \expecf{\tilde{x}_\calT }{ \expecf{(i,j) \in \binom{n}{2}  }{ \Paren{1 - \frac{ \Norm{ \pexpecf{\mu_\calT}{ x_i }  - \pexpecf{\mu_\calT}{ x_j } }}{ d_{i,j} } }^2 } } 
\end{equation}
in terms of $\OPT$, the cost of an optimal embedding.
Let $\Set{x_i^*}_{i\in [n]}$ be any embedding that obtains an objective value of $\OPT$.

Call a pair $(i,j)$  ``good'' (with respect to $\calT$) if it satisfies the following properties, where $\eps' = 16\sqrt{\OPT} + \eps$:
\begin{align}
   & \textbf{$d_{ij}$ is close to $\|x_i^* - x_j^*\|$:} \nonumber \\
   & \qquad \Paren{1-0.1}\norm{x_i^* - x_j^*}_2 \leq  d_{i,j} \leq \Paren{1+0.1} \norm{x_i^* - x_j^*}_2  \label{eqn:sdp-dist-vs-opt-dist} \\
   & \textbf{pseudo-deviations bounded by quantiles:} \nonumber \\
   & \qquad \expecf{\tilde{x}_{\calT} }{\pdev{\mu_{\calT} }{ x_i }^2}   \leq 10  \cdot \mathcal{Q}\Paren{ \Set{  \norm{x_i^* - x_k^* }^2  }_{k \in [n]} , \eps' }  \text{ and the same for $j$} \, ,
    \label{eqn:proof-translate} \\
    & \textbf{$x^*$ quantiles are small:} \nonumber \\
  & \qquad \frac{ \mathcal{Q}\Paren{ \Set{  \norm{x_i^* - x_k^* }^2  }_{k \in [n]} , \eps'  } }{\norm{x_i^* - x_j^*}^2}  + \frac{ \mathcal{Q}\Paren{ \Set{  \norm{x_j^* - x_k^* }^2  }_{k \in [n]} , \eps'  } }{\norm{x_i^* - x_j^*}^2} \leq 1.  
  \label{eqn:quantiles-small}
\end{align}
If  $(i,j)$ is not good, then we say it is ``bad''.

We put together the lemmas from the previous sections to prove the following claim at the end of this section:
\begin{claim}
\label{claim:few-bad-pairs}
    With probability at least $0.99$ over $\calT$, if $|\calT| \gg \log(1/\eps)/\eps$, there are at most $\bigO{\sqrt{\OPT} + \eps} \cdot \log \Delta \cdot \binom{n}{2}$ bad pairs.
\end{claim}

From now on, condition $\calT$ on the event that Claim~\ref{claim:few-bad-pairs} applies.
Using \cref{lem:sdp-val-small-subset} on the bad pairs, their contribution to the rounded cost, in expectation over $\tilde{x}_\calT$, is at most
\[
\E_{\tilde{x}_\calT} \E_{i,j \sim [n]} 1( (i,j) \text{ bad}) \cdot \Paren{1 - \frac{\| \pexpecf{\mu_\calT}{x_i - x_j} \|  }{d_{ij}}}^2 \leq \bigO{\sqrt{\OPT} + \eps} \cdot \log \Delta  \, .
\]
So, we can turn to the good pairs.

Observe, for any pair $(i,j)$, and any $\tilde{x}_\calT$, we have by Fact~\ref{fact:rounding-deviations},
% It now remains to bound pairs $(i,j)$ that are in $\calG$ and $i,j\in \calH$, i.e. $d_{i,j} = \Theta(\norm{x_i^* - x_j^*})$ and $i,j$ satisfy \cref{eqn:proof-translate}.
\begin{equation}
\label{eqn:expanding-single-term}
    \Paren{1 - \frac{ \Norm{ \pexpecf{\mu_\calT}{ x_i }  - \pexpecf{\mu_\calT}{ x_j } }}{ d_{i,j} } }^2 \leq 2 \pexpecf{\mu_\calT }{ \Paren{1 - \frac{  \norm{x_i - x_j} }  {d_{i,j}} }^2 }   + 2\Paren{ \frac{  \pdev{\mu_\calT}{x_i} +  \pdev{\mu_\calT}{x_j}  }{ d_{i,j} } }^2
\end{equation}
The first term on the right-hand side is the pair's contribution to the LP objective.
The second term we bound for good $(i,j)$, on average over $\tilde{x}_\calT$:
\begin{align*}
& \E_{\tilde{x}_\calT} \Paren{ \frac{  \pdev{\mu_\calT}{x_i} +  \pdev{\mu_\calT}{x_j}  }{ d_{i,j} } }^2\\
& \quad \leq \bigO{1} \cdot \frac {\mathcal{Q}\Paren{ \Set{  \norm{x_i^* - x_k^* }^2  }_{k \in [n]} , 16\sqrt{\OPT} +\eps  } + \mathcal{Q}\Paren{ \Set{  \norm{x_j^* - x_k^* }^2  }_{k \in [n]} , 16\sqrt{\OPT} +\eps  } } {\|x_i^* - x_j^*\|^2}
\end{align*}

At the end of this section, we use Lemma~\ref{lem:quantiles-with-subsets} to prove the following Claim:
\begin{claim}
\label{claim:avg-good-quantile-shrink}
    For $k > 2$,
    \begin{align*}
    & \E_{i,j \sim [n]} \left [ 1((i,j) \text{ good}) \cdot \frac {\mathcal{Q}\Paren{ \Set{  \norm{x_i^* - x_k^* }^2  }_{k \in [n]} , 16\sqrt{\OPT} +\eps  } + \mathcal{Q}\Paren{ \Set{  \norm{x_j^* - x_k^* }^2  }_{k \in [n]} , 16\sqrt{\OPT} +\eps  } } {\|x_i^* - x_j^*\|^2} \right ]\\
    & \leq \bigO{ \Paren{ \OPT^{1/k} +  \eps^{2/k}  }\cdot \log\Paren{\Delta} } \, .
    \end{align*}
    and, for $k = 1,2$, the same inequality holds with the right-hand side replaced with $\bigO{\Paren{\sqrt{\OPT} + \eps} \cdot \log \Delta}$ and $\bigO{\Paren{\sqrt{\OPT} + \eps} \cdot \log(1/\eps) \cdot \log \Delta}$, respectively.
\end{claim}

All together, for $k > 2$ we obtain for the good pairs
\[
\E_{\tilde{x}_\calT} \E_{i,j \sim [n]} 1((i,j) \text{ good}) \cdot \Paren{ 1 - \frac{ \|\pexpecf{\mu_\calT}{x_i - x_j} \| }{d_{ij}}}^2 \leq 2 \OPT + \bigO{\Paren{\OPT^{1/k} + \eps^{2/k}} \cdot \log \Delta} \, .
\]
and the analogous claim for $k=1,2$.
To obtain the theorem statement when $k > 2$, we rescale $\eps \leftarrow (\eps / \log(\Delta))^{k/2}$ (and analogously for $k=1,2$).
This gives us an algorithm using an LP of size $(n \Delta)^{\mathcal{O}(k^2 \log ( \log \Delta / \eps) \cdot (\log \Delta)^{k/2} / \eps^{k/2})}$.
We discuss in Section~\ref{sec:smaller-LP} how to prune this LP while maintaining the same guarantees.
\end{proof}

\paragraph{Proofs of Claims}
We turn to the proofs of Claims~\ref{claim:few-bad-pairs} and~\ref{claim:avg-good-quantile-shrink}.

\begin{proof}[Proof of Claim~\ref{claim:few-bad-pairs}]
  All but $\bigO{\OPT} n^2$ pairs $(i,j)$ satisfy \cref{eqn:sdp-dist-vs-opt-dist}, by Lemma~\ref{lem:relating-opt-to-dist-on-typical}.
  Combining Lemma~\ref{lem:var-redux-quantiles} and Lemma~\ref{lem:translating-bet-quantiles} shows that \cref{eqn:proof-translate} holds for all but $\bigO{\sqrt{\OPT} + \eps} n^2$ pairs, with probability at least $0.99$ over $\calT$, as long as $|\calT| \gg \log(1/\eps) / \eps$.
  And, applying Lemma~\ref{lem:quantiles-with-subsets} with $\delta = \Theta( \sqrt{\OPT} + \eps)$, we get that \cref{eqn:quantiles-small} holds for all but $\bigO{\sqrt{\OPT} + \eps} \cdot \log \Delta \cdot n^2$ pairs $i,j$.
\end{proof}

\begin{proof}[Proof of Claim~\ref{claim:avg-good-quantile-shrink}]
  Let $\calI$ be the set of pairs such \cref{eqn:quantiles-small} holds.
  It will be enough to prove the claim with $1((i,j) \text{ good})$ replaced with $1((i,j) \in \calI)$.
  For any $\delta>0$, applying \cref{lem:quantiles-with-subsets} with $\delta$ and $\eta= \eps'= 16\sqrt{\OPT}+\eps$ yields
  
  %for all but $\bigO{\delta \cdot \log(\Delta)n^2}$ pairs $(i, j) \in  \binom{n}{2}$, 
%\begin{equation*}
%\begin{split}
%    & \frac{\mathcal{Q}\Paren{ \Set{  \norm{x_i^* - x_k^* }^2  }_{k \in [n] } , \eps'  } + \mathcal{Q}\Paren{ \Set{  \norm{x_j^* - x_k^* }^2  }_{k \in [n] } , \eps'  }  }{ \norm{x_i^* - x_j^*}^2 } \\
%    & \hspace{0.5in}= \frac{\mathcal{Q}\Paren{ \Set{  \norm{x_i^* - x_k^* }  }_{k \in [n] } , \eps'  }^2 + \mathcal{Q}\Paren{ \Set{  \norm{x_j^* - x_k^* }  }_{k \in  [n]} ,\eps'  }^2  }{ \norm{x_i^* - x_j^*}^2 } \\
%    & \hspace{0.5in}\leq \frac{ \Paren{ \mathcal{Q}\Paren{ \Set{  \norm{x_i^* - x_k^* }  }_{k \in [n] } , \eps'  } + \mathcal{Q}\Paren{ \Set{  \norm{x_j^* - x_k^* }  }_{k \in  [n]} ,\eps'  } }^2 }{ \norm{x_i^* - x_j^*}^2  }  \\
%    & \hspace{0.5in} \leq \Paren{  \frac{c \Paren{ \sqrt{\OPT} +\eps } }{\delta } }^{2/k},
%\end{split}
%\end{equation*}
%for some fixed constant $c$,
%where the first inequality follows from recalling that each term in the fraction is at most $1$ since $(i,j)\in \calI$ and the second inequality follows from the conclusion of \cref{lem:quantiles-with-subsets}. 
\begin{equation}
\label{eqn:quantile-bound-main-proof}
\begin{split}
    &\probf{i,j \sim [n]}{  \frac{\mathcal{Q}\Paren{ \Set{  \norm{x_i^* - x_k^* }^2  }_{k \in [n]} , \eps'  } + \mathcal{Q}\Paren{ \Set{  \norm{x_j^* - x_k^* }^2  }_{k \in [n]} , \eps'  }  }{ \norm{x_i^* - x_j^*}^2 } > \Paren{ \frac{c(\sqrt{\OPT}+\eps)}{\delta} }^{2/k}  } \\
    & \leq \bigO{\delta \log(\Delta)}.
\end{split}
\end{equation}
for some fixed constant $c$.
First, suppose $k >2$.
For a non-negative random variable $x$ such that $x\leq 1$ with probability $1$,  $\expecf{}{x} = \int_{0}^{1} \prob{ x > t } dt$, so using \cref{eqn:quantile-bound-main-proof},

\begin{equation}
\label{eqn:bound-on-good-pairs-at-all-scales}
\begin{split}
& \expecf{ i,j \sim [n] }{ 1((i,j) \in \calI) \frac{\mathcal{Q}\Paren{ \Set{  \norm{x_i^* - x_k^* }^2  }_{k \in [n] } , \eps'  } + \mathcal{Q}\Paren{ \Set{  \norm{x_j^* - x_k^* }^2  }_{k \in [n] } , \eps'  }  }{ \norm{x_i^* - x_j^*}^2 }  } \\
& = \bigints_{0}^{1} \probf{ i,j \sim [n] }{ \frac{\mathcal{Q}\Paren{ \Set{  \norm{x_i^* - x_k^* }^2  }_{k \in [n] } , \eps'  } + \mathcal{Q}\Paren{ \Set{  \norm{x_j^* - x_k^* }^2  }_{k \in [n] } , \eps'  }  }{ \norm{x_i^* - x_j^*}^2 } > t } dt\\
& \leq   \bigO{\eps' \cdot \log(\Delta)} \cdot \int_{(\eps')^{2/k}}^1 \frac{dt}{t^{k/2}}  + \int_{0}^{ (\eps')^{2/k}} 1 \cdot dt
%& \leq \bigO{ \eps' \cdot \log(\Delta) } \int_{(\eps')^{2/k}}^{1} \frac{dt }{t^{k/2}}   + (\eps')^{2/k}\\
%& = \bigO{\eps'\cdot \log(\Delta) \Paren{\eps'}^{\frac{2}{k}(1-\frac{k}{2}) } }  +\eps'^{ 2/k} \\
\leq  \bigO{ \Paren{ \OPT^{1/k} +  \eps^{2/k}  }\cdot \log\Paren{\Delta} }. 
\end{split}
\end{equation}
Turning to the case that $k =1$, we can use the same calculation as \cref{eqn:bound-on-good-pairs-at-all-scales}, but now when we integrate $\int_{0}^1 dt / t^{1/2}$ we get $\bigO{1}$.
When $k=2$, the integral $\int_{\eps}^{1} dt/t $ gives $\bigO{\log 1/\eps}$.
\end{proof}

\section{Proofs of Distance-Quantile Lemmas}
\label{sec:distance-quantile-proofs}

In this section, we prove Lemma~\ref{lem:quantiles-with-subsets} through Lemma~\ref{lem:translating-bet-quantiles}, beginning with the simplest, Lemma~\ref{lem:relating-opt-to-dist-on-typical}.

\begin{proof}[Proof of Lemma~\ref{lem:relating-opt-to-dist-on-typical}]
By definition, we have
\begin{equation*}
\frac{1}{n^2} \sum_{i , j \in [n]} \Paren{ 1 - \frac{ \norm{x^*_i - x^*_j} }{d_{i,j} } }^2  = \OPT 
\end{equation*}
Then, by Markov's inequality, for any constant $c\ge 1$,
\begin{equation}
\label{eqn:avg-dist-ij}
\probf{i,j\sim [n]}{\Paren{ 1 - \frac{ \norm{x_i^* - x_j^*} }{d_{i,j} } }^2   \geq  c } \leq  \frac{\OPT}{c}
\end{equation}
Therefore, with probability at least $1-(\OPT/c)$, for a random pair $i,j\sim [n]$, 
\begin{equation}
\label{eqn:first-markov-bound}
    \Paren{ 1-\sqrt{c} } d_{i,j}  \leq  \norm{x_i^* -x_j^* } \leq \Paren{ 1+\sqrt{c} }  d_{i,j},
\end{equation}
and squaring both sides yields the first claim. 
Further, \cref{eqn:first-markov-bound} implies
\begin{equation*}
    \Paren{ 1-2\sqrt{c} } \norm{x_i^* -x_j^* }  \leq d_{i,j} \leq \Paren{ 1+2\sqrt{c} }  \norm{x_i^* -x_j^* }.
\end{equation*}
Next, recall the pseudo-distribution $\mu$ satisfies 
\begin{equation*}
     \pexpecf{\mu}{ \frac{1}{n^2} \sum_{i,j \in [n]} \Paren{ 1 - \frac{\norm{x_i - x_j }}{d_{i,j} } }^2  } \leq \OPT, 
\end{equation*}
Again by Markov's, with probability at least $1-\OPT/c$, for a random pair $i, j \sim [n]$, 
\begin{equation}
    \label{eqn:bound-constant-term}
    \pexpecf{\mu}{ \Paren{1- \frac{\norm{x_i - x_j} }{d_{i,j}}}^2 } \leq c
\end{equation}
% and therefore
Next, observe, for any $(i,j)$ satisfying \cref{eqn:bound-constant-term}, 
\begin{equation}
\begin{split}
    \pexpecf{\mu}{ \Paren{\frac{\norm{x_i-x_j} }{ d_{i,j} } }^2 } \leq 2 \Paren{ \pexpecf{\mu}{1} + \pexpecf{\mu}{ \Paren{1 - \frac{\norm{x_i-x_j} }{ d_{i,j} }  }^2 } } \leq 2+2c.
\end{split}
\end{equation}

Further, it follows from Jensen's inequality (going to local distributions and back) that 
\begin{equation*}
    \frac{1}{n^2} \sum_{i,j \in [n]} \Paren{ 1 - \frac{\pexpecf{\mu}{ \norm{x_i - x_j }} }{d_{i,j} } }^2   \leq \pexpecf{\mu}{ \frac{1}{n^2} \sum_{i,j \in [n]} \Paren{ 1 - \frac{\norm{x_i - x_j }}{d_{i,j} } }^2  } \leq \OPT ,
\end{equation*}
and thus with probability at least $1-\OPT/c$, for a random pair $i, j \sim [n]$, 
\begin{equation*}
     \Paren{ 1-\sqrt{c} } d_{i,j} \leq   \pexpecf{\mu}{ \norm{x_i  - x_j } } \leq  \Paren{ 1+\sqrt{c} }  d_{i,j} .
\end{equation*}
Further, this implies
\begin{equation*}
     \Paren{ 1-2\sqrt{c} }  \norm{x_i^* -x_j^* } \leq \pexpecf{\mu}{ \norm{x_i  - x_j } }  \leq  \Paren{ 1+2\sqrt{c} }  \norm{x_i^* -x_j^* } ,
\end{equation*}
which yields the second claim. We can complete the proof by observing that 
\begin{equation*}
    \begin{split}
        \pexpecf{\mu}{ \Paren{\frac{\norm{x_i-x_j} }{ d_{i,j} } }^2 } \geq \Paren{  \frac{\pexpecf{\mu}{\norm{ x_i - x_j} } }{d_{i,j}} }^2 \geq \Paren{\frac{1}{1+\sqrt{c}} }^2 \geq 1-2\sqrt{c}.
    \end{split}
\end{equation*}

\end{proof}

% We now have the following simple corollary relating the distances of the optimal target embedding to the expected distances under the pseudo-distribution $\mu$:
% \begin{corollary}[Optimal Embedding vs SDP solution]
% \label{cor:opt-vs-sdp-embedding}
% Let $\eta = \OPT /n^2$ and let $\Set{x_1^*, x_2^*, \ldots, x_n^*}$ be the optimal embedding of the points. Let $\mu$ be any pseudo-distribution that minimizes the objective $\sum_{ i,j\in[n] } \Abs{ 1 - \frac{\abs{x_i - x_j} }{d_{i,j} } } $. Then, for any constant $c<1/2$, with probability at least $1-2\eta/c$, for a random pair $i,j \sim [n]$, 
% \begin{equation*}
%     (1 - 2c) \abs{ x_i^* - x_j^* } \leq \pexpecf{\mu}{ \abs{ x_i - x_j } } \leq (1+2c) \abs{ x_i^* - x_j^* }
% \end{equation*}
% \end{corollary}
% \begin{proof}
% By union bound, both conclusions of Lemma~\ref{lem:relating-opt-to-dist-on-typical} hold with probability at least $1-2c\eta$. Then,
% \begin{equation*}
%     \pexpecf{\mu}{ \abs{ x_i - x_j }} \leq \frac{d_{i,j}}{1-c} \leq \frac{1+c}{1-c} \abs{x_i^* - x_j^*}  \leq (1+2c)\abs{x_i^* - x_j^*} .
% \end{equation*}
% And similarly, 
% \begin{equation*}
%     \pexpecf{\mu}{ \abs{ x_i - x_j }} \geq \frac{d_{i,j}}{1+c} \geq \frac{1-c}{1+c} \abs{x_i^* -x_j^*} \geq (1-2c) \abs{x_i^* -x_j^*} ,
% \end{equation*}
% which concludes the proof.
% \end{proof}

\subsection{Quantiles of pair-wise distances on the line -- proof of Lemma~\ref{lem:quantiles-with-subsets}}
\label{sec:quantiles-of-pairwise-distances}

In this section, we show that quantiles of distances between points in $\R^k$ are well-behaved. 

% \ainesh{this is no longer on the line is it?}
% \Snote{fixed}

\begin{proof}[Proof of Lemma~\ref{lem:quantiles-with-subsets}]

We begin by partitioning the pair-wise distances that are at least $1$ into geometrically increasing level sets in the range $[1,2\Delta]$ as follows: let $\Set{ \calP_0, \calP_1, \ldots \calP_{\log(\Delta)} }$ be a partition of $\calS_1 \times \calS_2$  such that for each $T \in [-1,\log(\Delta)]$,
\begin{equation}
    \calP_T = \Set{ (i,j)  \hspace{0.1in} \big\vert \hspace{0.1in}   \norm{x_i^* - x_j^* } \in (2^T , 2^{T+1}], i , j \in [n]  }. 
\end{equation}
% Note, we can w.l.o.g. assume that all the pair-wise distances are strictly greater than $1$, since there are at most $n$ pairs that can be at distance exactly $1$ and we can add them to the set of pairs we drop. 
%\ainesh{some corner case thing to avoid here, not really an issue coz there can be at most $n$ distances that are $1$. }

Next, for each $T \in [0, \log(\Delta)]$, we
fix a $(2^T/2)$-cover, denoted by $\calN_{k, \Delta, 2^T}$. Thus, $\min_{x,y \in \calN_{k, \Delta, 2^T}} \|x-y\| > 2^T/2$, and for any $x \in [-\Delta, \Delta]^{\otimes k}$  $\min_{z \in \calN_{k, \Delta, 2^T}} \|x-z\| \leq 2^T/2$. We now partition the space $[-\Delta, \Delta]^{\otimes k}$ by mapping each $x \in [-\Delta, \Delta]^{\otimes k}$ to its nearest $z \in \calN_{k, \Delta, 2^T}$, breaking ties arbitrarily. Let $\calG_T$ denote this partition, and note that each cluster $S \in \calG_T$ has diameter at most $2^T$. It follows that all distances at scale $\calP_T$ arise from points $x_i^*, x_j^*$ that lie in distinct clusters of $\calG_T$.  
Similarly, we define a finer partition of $\calG_T$, constructed via a smaller $\eps$-cover with $ \eps = \eta \cdot 2^T / \delta$, which we denote by $\calF_{\eta, T}$. For each $\calB \in \calF_{\eta , T} \cup \calG_T$ let $\calC(\calB)$ denote the number of points such that $i \in [n]$ and $x_i^*\in \calB$.
If for any cluster $\calB \in  \calF_{\eta , T}$, i.e. the finer partition, $\calC(\calB) \leq \bigO{ \eta^{k}  n }$, we discard all pairs $(i,j) \in \calP_T$ that involve any $x_i^*$ such that $x_i^* \in \calB$ or $x_j^* \in \calB$.

Observe, for the remaining pairs  $(i,j) \in \calP_T$, we know that $x_i^*$ and $x_j^*$ land in distinct finer clusters from $\calF_{\eta , T}$ that each have at least $c\eta^{k} n$ points (for a fixed constant $c$) $x^*_k$ such that $\norm{x_i^* - x_k^* } ,\norm{x_j^* - x_k^* } \leq \frac{2 \eta}{\delta} \norm{x_i^* - x_j^*}$. Therefore, 

\begin{equation}
    \calQ\Paren{ \Set{ \Norm{x_i^* - x_k^* } }_{k \in [n] } , c \eta^{k }   } + \calQ\Paren{ \Set{ \Norm{ x_j^* - x_k^* }}_{k \in [n] } ,  c \eta^{k}   } \leq \Paren{ \frac{ 4\eta }{\delta}  }\Norm{x_i^* - x_j^* },
\end{equation}
where the inequality follows from recalling that $(i,j) \in \calP_T$.

It remains to show that the above procedure discards at most $\bigO{ \delta^{k}  n^2 }$ pairs. Recall,  $\calG_T =\Set{ \calB_\ell }_{\ell \in [\abs{\calN_{k,\Delta, 2^T} } ]}$ is the set of clusters with diameter $2^T$ that cover $[-\Delta, \Delta]^{\otimes k}$. %For each such cluster $\calB_\ell $, similarly let $\calC(\calB_\ell)$ be the number of points that land in it. 
Further, note that for each $\calB_\ell \in \calG_{T}$, by the separation property of the cover and the doubling dimension of $\R^k$, we have that at most  $\Paren{c'\delta/\eta}^{k}$ clusters in $\calF_{\eta, T}$ intersect (and cover) $\calB_\ell$, for a fixed constant $c'$. 

For a fixed cluster $\calB_\ell \in \calG_T$, we consider the number of pairs $(i,j)$ dropped because either $x_i^*$ or $x_j^*$ lies in a sub-cluster $\calS \in \calF_{ \eta^k , T}$ intersecting $\calB_\ell$ with at most $c \eta^k n$ points; call such a cluster $\calS$ \textit{light}. Next, observe that any $x_i^*$ in a light sub-cluster, for every  distance $(x_i^*,x_j^*)$ in $\calP_T$ the point $x_j^*$ can be in one of at most $\bigO{ 2^{k} }$ clusters $\calB_{\ell} \in \calG_T$ since the ball of radius $2^{2T+1}$ around $x_i^*$ is covered by  $\bigO{ 2^{k} }$ clusters in $\calG_T$ of radius $2^{T}/2$ (by the separation property of the cover). Therefore, the total number of such pairs is at most $\sum_{z \in [ \mathcal{O}(2^{k}) ] } \calC\Paren{ \calB_{\ell_z} } $. Therefore, the total number of pairs dropped for a fixed cluster $\calB_\ell$ is at most  

% Each point in such a sub-interval can contribute to at most $\calC_{\calS_2}(\calI_{\ell-2}) + \calC_{\calS_2}(\calI_{\ell-1}) +  \calC_{\calS_2}(\calI_{\ell+1}) +  \calC_{\calS_2}(\calI_{\ell+2})$ pairs, since the distances are restricted to be $(2^T , 2^{T+1}]$. Therefore, the total number of pairs dropped such that $i \in \calS_1$, and $j \in \calS_2$ is at most 
\begin{equation*}
    \underbrace{ \bigO{ \frac{\delta}{\eta}  }^{k} }_{\textrm{number of light clusters}} \cdot \underbrace{ \bigO{ \eta^{  k }  n } }_{\textrm{ $x_i^*$ in light clusters }} \cdot \underbrace{  \sum_{z \in [ \mathcal{O}(2^{k}) ] } \calC\Paren{ \calB_{\ell_z} }  }_{  \textrm{number of $x_j^*$s in any neighboring cluster }}.  
\end{equation*}
Summing over all $\ell \in  \abs{ \calN_{k ,\Delta, \eps} }$, the total number of such pairs dropped can be bounded as follows:
\begin{equation*}
     \bigO{ \delta^{k } \cdot 2^{k } n^2 } \leq \bigO{ \delta^{k } n^2}
\end{equation*}
% Similarly, we can bound the number of pairs we drop such that $j\in \calS_2$ by 
% \begin{equation*}
%     \Paren{ \frac{4\delta }{(\eta_1 + \eta_2)} } \cdot \Paren{\eta_2 \abs{\calS_2} }  \cdot \abs{\calS_1} .
% \end{equation*}
 Summing over the $\log(\Delta)$ levels, and substituting $\eta \gets \eta^{1/k} $ and $\delta\gets \delta^{ 1/k}$ concludes the proof.

% Since we drop every sub-interval in $\calG_{\eps \cdot T/\delta}$ with less than $\eps n $ points, we drop at most $\frac{\delta}{ \eps } \cdot \eps n  = \delta n$ from $\calI_\ell$. Since pair-wise distances lie in adjacent intervals, it suffices to focus on pairs of consecutive intervals.
% For any $\ell \in [1, 2\Delta/2^T -1]$, we can bound the number of pairs we drop between $\calI_\ell$ and $\calI_{\ell+1}$ as follows: 
% \begin{equation*}
%     \frac{\delta}{\eps} \cdot \eps n \cdot \calC(I_{\ell+1}) +  \frac{\delta}{\eps} \cdot \eps n \cdot \calC(I_{\ell}) = \delta n \Paren{ \calC(\calI_\ell) + \calC(\calI_{\ell+1} ) } .
% \end{equation*}
% points from $\calC_\ell$. 

% Summing over all $\ell \in [1, 2\Delta/2^T]$, we drop 
% \begin{equation}
%    \delta n \cdot \Paren{  \sum_{\ell \in [1, 2\Delta / 2^T-1]}  \calC(\calI_\ell) + \calC(\calI_{\ell+1})  } \leq 2\delta n \cdot  \calC\Paren{\calG_T } = 2\delta \cdot n^2.
% \end{equation}
% Summing over the $\log(\Delta)$ different scales, we drop at most $2\delta \log(\Delta) n^2$ pairs overall, which completes the proof. 

\end{proof}

\subsection{Translating between pseudo-expected distances and optimal quantiles}
\label{sec:translating}

In this section, we show that we can bound the pseudo-deviation of each indeterminate in terms of the quantiles of the pair-wise distances in the optimal embedding corresponding to that indeterminate. However, in this translation, we lose a $\sqrt{\OPT}$ factor, since we can only relate the $\sqrt{\OPT}$-th quantiles of the pseudo-expected distances to the distances in an optimal embedding.

\begin{proof}[Proof of \cref{lem:translating-bet-quantiles}]

%  It follows from \cref{lem:var-redux-quantiles} that as long as $\calT \geq \Omega\Paren{ \log(1/\eps)/\eps }$, for each $i \in [n]$,  with probability at least $1-\eps$, 
% \begin{equation}
% \label{eqn:dev-to-pseudo-quantile}
% \expecf{  \tilde{x}_\calT }{    \pdev{\mu_\calT}{ x_i }^2 }  \leq    \calQ\Paren{ \Set{ \pexpecf{\mu}{ \norm{x_i - x_j }^2 }  }_{j \in [n]}, \eps  } ,
% \end{equation}
% Let $\calB_1 \subset [n]$ be the set of points that do not satisfy \cref{eqn:dev-to-pseudo-quantile}. Using a standard concentration bound, we know $\abs{\calB_1} \leq \bigO{\eps}$ with probability at least $0.99$. For each $i \in [n] \setminus \calB_1$, in expectation over the random assignments to each point in $\tilde{x}_\calT$, (after fixing the set of indices $\calT$), we have
% \begin{equation*}
%     \expecf{  \tilde{x}_\calT }{    \pdev{\mu_\calT}{ x_i }^2 }  \leq    \calQ\Paren{ \Set{ \pexpecf{\mu}{ \norm{x_i - x_j }^2 }  }_{j \in [n]}, \eps  } .
% \end{equation*}
For each $i \in [n]$, let $\eps_i$ be such that 
\begin{equation*}
    \mathcal{Q}\Paren{ \Set{ \pexpecf{\mu }{ \norm{x_i - x_k }^2 } }_{k \in [n]} ,   \eps } \leq 10  \cdot \mathcal{Q}\Paren{ \Set{  \norm{x_i^* - x_k^* }^2  }_{k \in [n]} ,  \eps +  \eps_i  }  .
\end{equation*}
It follows from \cref{lem:relating-opt-to-dist-on-typical}, with $c=1/16$ that
\begin{equation}
\label{eqn:bound-on-sum-eps_i}
\frac{1}{n} \sum_{i \in [n]} { \eps_i  } \leq  16 \cdot \OPT,
\end{equation}
since each $\eps_i$ contributes at least $\eps_i \cdot n$ pairs such that $(i,k)$ does not satisfy 
\begin{equation*}
    \Paren{1-0.75}\norm{x_i^* - x_k^*}^2_2 \leq  \pexpecf{\mu}{ \norm{ x_i - x_k }_2^2 } \leq 2.5 \norm{x_i^* - x_k^*}_2^2,
\end{equation*}
and we know that there are at most $16 \OPT \cdot n^2$ such pairs. 
% \ainesh{technically, we need this statement for conditioned pseudo-distributions, but this is fine because we should have a lemma that says conditioned pseudo-distributions don't increase the objective value. }
Thus applying Markov's to \eqref{eqn:bound-on-sum-eps_i},
\begin{equation}
\probf{i \sim [n]}{ \eps_i > 16\sqrt{\OPT} } \leq \sqrt{\OPT} .
\end{equation}
Therefore, for all but $\sqrt{\OPT}$-fraction of points, we have
\begin{equation}
\label{eqn:pseudo-quantile-to-opt-quantile}
    \mathcal{Q}\Paren{ \Set{ \pexpecf{\mu }{ \norm{x_i - x_k }^2 } }_{k \in [n]} ,  \eps } \leq 10  \cdot \mathcal{Q}\Paren{ \Set{  \norm{x_i^* - x_k^* }^2  }_{k \in [n]} , 16\sqrt{\OPT} +\eps  }  .
\end{equation}
Let $\calB_2$ be the set of points that do not satisfy \cref{eqn:pseudo-quantile-to-opt-quantile}. Observe, $\abs{ \calB_1 +\calB_2} \leq \bigO{\sqrt{\OPT} +\eps} n^2$, which concludes the proof.
\end{proof}

\section{Remaining Proofs: Variance Reduction and Discretization}
\label{sec:deferred-proofs-var-reduction-discretization}

First, we prove the simple Fact~\ref{fact:rounding-deviations}, then we turn to the remaining key lemmas: the discretization Lemma~\ref{lem:aspect-ratio} and the variance-reduction Lemma~\ref{lem:var-redux-quantiles}.

\begin{proof}[Proof of Fact~\ref{fact:rounding-deviations}]
We have
\begin{equation}
\label{eqn:expanding-single-term-part2}
\begin{split}
   \Paren{1 - \frac{ \Norm{ \E X  - \E Y }}{ d } }^2 
    & =  \Paren{1 - \frac{ \Norm{ \E X  - \E Y} + \E \|X - Y\| - \E \|X - Y\| }{ d } }^2 \\
    & \leq 2\Paren{1 - \frac{ \E \| X - Y \| } {d} }^2  +  2\Paren{ \frac{  \Norm{ \E X - \E Y } - \E \|X - Y\|   }{ d }    }^2 \\
    & \leq 2 \E \Paren{1 - \frac{ \|X - Y\|}  {d} }^2  +  2\Paren{ \frac{  \Norm{ \E X - \E Y } - \E \|X - Y\|   }{ d }    }^2\\
\end{split}
\end{equation}
where we used the inequality $(a+b)^2 \leq 2(a^2 + b^2)$ and Jensen's inequality.
For the last term above, we have
\begin{align*}
\| \E X - \E Y \| & = \E \| (\E X - \E Y) + (X - Y) - (X - Y) \|\\
& \leq \E \|X - \E X\| + \E \|Y - \E Y\| + \E \|X - Y\|
\end{align*}
and so therefore
\[
\|\E X - \E Y\| - \E \|X - Y \| \leq \E \|X - \E X\| + \E \|Y - \E Y\| \, .
\]
A similar argument shows
\[
\E \|X - Y\| - \|\E X - \E Y\| \leq \E \|X - \E X\| + \E \|Y - \E Y\|
\]
which completes the proof.
\end{proof}

\subsection{Variance Reduction via Conditioning}
\label{sec:varredux-proof}

In this section, we show that the pseudo-deviation of the variables $x_i$ is related to the $\eps$-th quantile of pseudo-expected distances. In particular, we provide a proof of Lemma~\ref{lem:var-redux-quantiles}.

% \begin{theorem}[Conditioning Reduces Variance when $\OPT$ is small]
% \label{thm:top-level-conditioning}
% Given $0<\eps < 0$, $1\leq \Delta$, let $\mu$ be a degree-$\bigO{ 1/\eps}$ pseudo-distribution consistent with the optimization program in Equation~\eqref{eqn:optimization-program}. Let $\calT\subset [n]$ be a random subset of size $1/\eps$ and let $\Set{ \tilde{x}_j }_{j \in \calT}$ be drawn from the pseudo-marginals (see Definition~\ref{def:pseudo-marginals}). 
% Let $\mu'$ be the pseudo-distribution obtained by conditioning on the event $\Set{ x_j = \tilde{x}_j }_{j\in \calT}$. \Snote{need to specify that $x_i^*$ are drawn according to marginal distribution $\{x_i\}_{i \in T}$} Then, with probability at least $1-1/\poly(n)$,
% \begin{equation*}
%     \expecf{\calT}{ \frac{1}{n} \sum_{i \in [n]}   \pdev{\mu'}{x_i} } \leq \eps \Delta + \ldots . 
% \end{equation*}
% \ainesh{wip.}
% \end{theorem}

We begin by showing that conditioning on a subset of points reduces the pseudo-deviation of any point to the distance between this point and the closest conditioned point. Formally,

\begin{lemma}[Pseudo-Deviation to Closest Conditioning]
\label{lem:pseudo-deviation-to-closest-conditioning}
Let $\calT \subset [n]$ be any arbitrary subset of size $\tau \in[n]$, and let $\mu$ be a degree-$2\tau$ pseudo-distribution over the vector-valued indeterminate $x \in \mathbb{R}^n$. Let $\tilde{x}_\calT \in \calN_{k,\Delta/\eps,1}^{\abs{\calT} }$ be a  draw from the local distribution $\Set{x_\calT}$ and $\mu_{\calT}$ be the pseudo-distribution obtained by conditioning on the event $\zeta_\calT : = \Set{ x_j = \tilde{x}_j }_{j \in \calT}$. Then, for all $i \in [n]$,  
\begin{equation*}
 \expecf{ \tilde{x}_\calT }{ \pdev{\mu_\calT}{x_i }^2 } \leq 2 \cdot \min_{ j \in \calT }  \pexpecf{\mu }{\Norm{ x_i - x_j }^2  }  .    
\end{equation*}  
\end{lemma}
\begin{proof}
% Let $\Set{ \tilde{x}_{j_1}, \tilde{x}_{j_2}, \ldots , \tilde{x}_{j_\tau} }$ be an ordering of the indeterminates we condition on. Then, for each $t \in [\tau]$, let $\mu_t$ be the pseudo-distribution obtained by conditioning on $x_{j_{t'}} = \tilde{x}_{j_{t'}}$ for all $t' \leq t$. 
Since $\mu$ has degree at least $2\tau$, any pseudo-distribution over $\tau$ variables can be treated as a true local distribution. Then, for any fixed $j \in \calT$,

% recall $\pdev{\mu_t}{x_i }= \pdev{\mu_{t-1}}{ x_i \vert x_{j_t}= \tilde{x}_{j_t} }$ and thus

\begin{equation}
\begin{split}
   \expecf{ \tilde{x}_\calT }{ \pdev{\mu_{\calT}  }{ x_i  }^2 } & = \expecf{ \tilde{x}_\calT  }{ \Paren{  \pexpecf{\mu }{ \Norm{ x_i - \pexpecf{\mu}{ x_i \vert \zeta_\calT } } 
 \Big\vert \zeta_\calT } }^2 }\\
 & =   \expecf{  \tilde{x}_\calT  }{ \Paren{ \pexpecf{\mu }{ \Norm{ x_i - x_{j} + \Paren{ x_j - \pexpecf{\mu}{ x_i \vert \zeta_\calT  } } }  \Big\vert \zeta_\tau } }^2 } \\
 &\leq 2 \expecf{ \tilde{x}_\calT  }{ \Paren{  \pexpecf{\mu }{ \Norm{ x_i - x_{j}   }  \Big\vert \zeta_\calT } }^2 }    + 2 \expecf{ \tilde{x}_\calT   }{ \Paren{ \pexpecf{\mu }{  \Norm{ x_j - \pexpecf{\mu}{ x_i \vert \zeta_\calT  } }  \Big\vert \zeta_\calT   } }^2  }   \\
 & \leq 2  \expecf{ \tilde{x}_\calT  }{    \pexpecf{\mu }{ \Norm{ x_i - x_{j}   }^2  \Big\vert \zeta_\calT }   }  + 2 \expecf{\tilde{x}_\calT }{ \Paren{ \pexpecf{\mu }{  \Norm{  \pexpecf{\mu}{ x_j - x_i \vert \zeta_\calT  } }  \Big\vert \zeta_\calT   } }^2 }\\
 & \leq 2  \pexpecf{\mu}{\norm{ x_i - x_j}^2 } + 2\expecf{ \tilde{x_\calT} }{ \Paren{ \pexpecf{\mu}{\norm{ x_j - x_i } \vert \zeta_\calT  }  }^2 } \\
 & \leq 4 \pexpecf{\mu}{\norm{ x_i - x_j}^2 }, 
\end{split}
\end{equation}
% \begin{equation}
% \begin{split}
%    \expecf{ \tilde{x}_\calT }{ \pdev{\mu_{\calT}  }{ x_i  } } & = \expecf{ \tilde{x}_\calT  }{ \pexpecf{\mu }{ \Abs{ x_i - \pexpecf{\mu}{ x_i \vert \zeta_\calT } } 
%  \Big\vert \zeta_\calT } }\\
%  & = \expecf{ \tilde{x}_\calT  }{ \pexpecf{\mu }{ \Abs{ x_i - x_{j} + \Paren{ x_j - \pexpecf{\mu}{ x_i \vert \zeta_\calT  } } }  \Big\vert \zeta_\tau } } \\
%  &\leq \expecf{ \tilde{x}_\calT  }{ \pexpecf{\mu }{ \Abs{ x_i - x_{j}   }  \Big\vert \zeta_\calT } }   + \expecf{ \tilde{x}_\calT   }{ \pexpecf{\mu }{  \Abs{ x_j - \pexpecf{\mu}{ x_i \vert \zeta_\calT  } }  \Big\vert \zeta_\calT   }  }   \\
%  & =  \pexpecf{\mu}{\abs{ x_i - x_j} } + \expecf{\zeta_\calT }{ \pexpecf{\mu }{  \Abs{  \pexpecf{\mu}{ x_j - x_i \vert \zeta_\calT  } }  \Big\vert \zeta_\calT   }  }\\
%  & =  \pexpecf{\mu}{\abs{ x_i - x_j} } + \expecf{\zeta_\calT }{\Abs{ \pexpecf{\mu}{ x_j - x_i \vert \zeta_\tau  }  }} \\
%  & \leq 2 \pexpecf{\mu}{\abs{ x_i - x_j} }, 
% \end{split}
% \end{equation}
where the first inequality follows from triangle inequality, the second follows from Jensen's inequality, combined with the law of conditional expectation.
Repeating the above inequality for all $j \in \calT$, and taking the minimum yields the claim.
\end{proof}

% \begin{lemma}[Closest Conditioning is small]
% Given $0<\eps < 0$, $1\leq \Delta$, let $\mu$ be a degree-$\bigO{ 1/\eps}$ pseudo-distribution consistent with the optimization program in Equation~\eqref{eqn:optimization-program}. Let $\calT\subset [n]$ be a random subset of size $1/\eps$ and let $\Set{ x_j = \tilde{x}_j }_{j \in \calT}$ be drawn from the pseudo-marginals (see Definition~\ref{def:pseudo-marginals}). Let $\mu_\calT$ be the pseudo-distribution obtained by conditioning on $\zeta_\calT = \Set{ x_j = \tilde{x}_j }_{j \in \calT}$. 
% Then,
% \begin{equation*}
%     \expecf{ \zeta_\calT }{ \expecf{ i \sim [n] }{  \pdev{\mu_\calT}{ x_i } } } \leq \expecf{i \sim [n] }{ \calQ\Paren{ \Set{ \abs{x_i^* - x_j^*} }_{j \in [n]}, 1/t  } }  + \frac{1}{t}\expecf{ i \sim [n]}{  \pdev{\mu_\calT }{x_i} }, 
% \end{equation*}
% where $\calQ\Paren{ \calS, q }$ denotes the $q$-th quantile of the ordered set $\calS$. \ainesh{this translation requires that $\OPT$ is small.}
% \end{lemma}

We can now complete the proof of Lemma~\ref{lem:var-redux-quantiles}. 

\begin{proof}[Proof of Lemma~\ref{lem:var-redux-quantiles}]
Let $\zeta_{i, \calT}$ be the event that the set $\calT$ contains $j$ such that $$ \pexpecf{}{ \norm{ x_i  -x_j }^2 }  \leq \calQ\Paren{  \Set{ \pexpecf{}{ \norm{x_i - x_j }^2 }   }_{ j \in [n] } ,  \eps}.$$  Then,
\begin{equation}
\label{eqn:bounding-the-deviation}
\begin{split}
    \expecf{ \zeta_{\calT} }{    \pdev{\mu_\calT}{ x_i }^2 } 
    & =   \expecf{  \zeta_\calT }{  \pdev{\mu_\calT}{ x_i }^2 \Paren{ \indic{\zeta_{i, \calT} } + 1 - \indic{\zeta_{i, \calT} } }    }   \\
    & \leq 2  \Paren{  \min_{j \in \calT} \pexpecf{\mu}{\norm{x_i - x_j}^2}  }\cdot \indic{\zeta_{i, \calT} }         + \expecf{  \zeta_\calT }{ \Paren{  1 - \indic{\zeta_{i, \calT} }  } \pdev{\mu_\calT}{ x_i }  }     \\
    & \leq 2  \calQ\Paren{  \Set{ \pexpecf{\mu }{ \norm{x_i - x_j }^2 }   }_{ j \in [n] } ,  \eps }  \cdot \indic{\zeta_{i, \calT} }         +  \Paren{  1 - \indic{\zeta_{i, \calT} }  } \cdot \pdev{\mu_\calT}{ x_i }   
\end{split}
\end{equation}
where the first inequality uses Lemma~\ref{lem:pseudo-deviation-to-closest-conditioning} and the second uses that when $\indic{\zeta_{i,\calT}}=1$, $$\min_{j \in \calT} \expecf{\mu}{\norm{x_i - x_j}^2} \leq \calQ\Paren{ \Set{ \expecf{\mu}{ \norm{x_i - x_j}^2 }}_{j\in[n]}, \eps }.$$ 
Next, observe 
\begin{equation}
    \Pr\left[ \zeta_{i \calT} = 0 \right] \leq \Paren{1- \eps}^{\tau} \leq \bigO{\delta}. 
\end{equation}
Conditioning on the event that $\zeta_{i \calT} =1$, it follows from Equation~\eqref{eqn:bounding-the-deviation}  that 
 with probability at least $1-\delta$ over the choice of $\calT$, 
\begin{equation*}
     \expecf{  \tilde{x}_\calT }{    \pdev{\mu_\calT}{ x_i }^2 }  \leq \calQ\Paren{ \Set{ \pexpecf{\mu}{ \norm{x_i - x_j }^2 }  }_{j \in [n]}, \eps  } ,
\end{equation*}
which concludes the proof.

\end{proof}

\subsection{Discretization}
\label{sec:discretization}

\begin{proof}[Proof of Lemma~\ref{lem:aspect-ratio}]
Let $\Set{x_i^*}_{i \in [n]}$ be an embedding with cost $\OPT$. First, using Jensen's inequality:

\begin{equation}
    \begin{split}
       1 \geq \OPT &= \expecf{i,j}{\left(1-\frac{\|x_i^* - x_j^*\|}{d_{i,j}}\right)^2} \\
     & \geq 1 - 2 \expecf{i,j}{\frac{\|x_i^* - x_j^*\|}{d_{i,j}}} + \left(\expecf{i,j}{\frac{\|x_i^* - x_j^*\|}{d_{i,j}}}\right)^2 \\
    \end{split}
\end{equation}
which implies that 
\begin{equation}
 \expecf{i,j}{\frac{\|x_i^* - x_j^*\|}{\Delta}} \leq   \expecf{i,j}{\frac{\|x_i^* - x_j^*\|}{d_{i,j}}}  \leq 2
\end{equation}
and therefore $\expecf{i,j}{\|x_i^* - x_j^*\|} \leq 2 \Delta$. Letting $c \in \{x_i^*\}_{i \in [n]}$ be the center of $ \{x_i^*\}_{i \in [n]}$, i.e. a point such that
\begin{equation*}
    c=\arg\min_{i \in [n]} \sum_{j \in [n]} \|x_i^* - x_j^*\|
\end{equation*} 
It follows for such a choice of $c$ that $\expecf{i \sim [n]}{\|c-x_i^*\|} \leq 2 \Delta$. Now define the modified point set $\{x_i\}_{i \in [n]}$ as the image of the projection of $\{x_i^*\}_{i \in [n]}$ onto the metric ball $\mathcal{B}\Paren{ c,\frac{10}{\eps} \Delta} $. By Markov's inequality, we have $\Abs{ X^* \cap \mathcal{B}\Paren{ c,\frac{10}{\eps} \Delta} } \geq (1-\frac{\eps}{5})n$. Thus, assuming $\eps$ is smaller than some constant, we have
\[ \left| \left\{ (i,j) \; : \; \|x_i^* - x_j^*\| \neq \|x_i - x_j\|\right\}\right| \leq \frac{\eps^2}{25} n^2 + \frac{\eps}{5}n^2 \leq \frac{\eps}{4}n^2 \]
Moreover, since the mapping $X^* \mapsto \mathcal{B}\Paren{ c,\frac{10}{\eps} \Delta}  $ was a contraction, for each such above pair $i,j$ we have:
\[\left(1-\frac{\|x_i- x_j\|}{d_{i,j}}\right)^2 -\left(1-\frac{\|x_i^* - x_j^*\|}{d_{i,j}}\right)^2 \leq 1  \] 
from which if follows that $\frac{1}{n^2} \sum_{i,j} \left( 1- \frac{\norm{x_i - x_j}  }{d_{i,j} } \right)^2 \leq \OPT + \eps/4$. Next, we fix an $\eps$-net $\mathcal{N}$ over $\mathcal{B}\Paren{ c,\frac{10}{\eps} \Delta} $;  by standard volume arguments, such nets exist with size $|\mathcal{N}| \leq \bigO{(\frac{\Delta}{\eps})^k}$. We then map each point $x_i$ to the nearest net point in $\mathcal{N}$, and let $\{y_i\}_{i \in [n]}$ be the resulting mapping. For all $i,j \in \binom{n}{2}$, we have
\begin{equation*}
    \begin{split}
        \left| \left(1-\frac{\|x_i- x_j\|}{d_{i,j}}\right)^2 -\left(1-\frac{\|y_i - y_j\|}{d_{i,j}}\right)^2 \right| &\leq \left| \left(1-\frac{\|x_i- x_j\|}{d_{i,j}}\right)^2 -\left(1-\frac{\|x_i - x_j\|}{d_{i,j}} \pm \eps\right)^2 \right| \\
        & \leq  2 \eps \left(1-\frac{\|x_i- x_j\|}{d_{i,j}}\right)^2  + \eps^2
    \end{split}
\end{equation*}
Thus 
\begin{equation*}
\begin{split}
   \frac{1}{n^2}   \sum_{i,j} \Paren{ 1- \frac{\norm{y_i - y_j}  }{d_{i,j} } }^2 &\leq    \frac{1}{n^2}\sum_{i,j} (1+2\eps)\left( 1- \frac{\norm{x_i - x_j}  }{d_{i,j} } \right)^2 + \eps^2\\
   &\leq (1+2 \eps) \left(\OPT + \eps/4\right)  + \eps^2 \leq \OPT + 3\eps
\end{split}
\end{equation*}
which completes the proof with $\{y_i\}_{i\in [n]}$ being the desired mapping after a rescaling of $\eps$. %\raj{Rescale $\eps$ above, deal with case of $x_i = x_j$} 
%The max distance in the target space is at most $\Delta \log(\log(\Delta))$ (\ainesh{cite demaine paper)} and \ainesh{add argument for grid of granularity $\eps$ contribution additive error $\eps n^2$. }
\end{proof}

\section{Dimension Reduction for Kamada-Kawai}
\label{sec:dim-reduction}

%\ainesh{verify the gap in this analysis.}\ainesh{done}
\begin{lemma}[Dimension Reduction for Kamada-Kawai]
    Let $\{d_{i,j}\}_{i \neq j \in [n]}$ be positive, and for $k \in \N$ let
    \[
    \OPT_k = \min_{x_1,\ldots,x_n \in \R^k} \frac 1 {n^2} \sum_{i \neq j \in [n]} \Paren{ 1- \frac{\|x_i - x_j\|}{d_{i,j}}}^2 \, .
    \]
    For every $\eps > 0$ and $k$,
    \[
    \OPT_{1/\eps} \leq \OPT_k + \bigO{\eps} \, .
    \]
\end{lemma}
\begin{proof}
    Since $\OPT_k$ is monotonically non-increasing in $k$, if $k \leq 1/\eps$ then the statement is trivial.
    So we can assume $k > 1/\eps$. Let $S \in \R^{ 1/\eps  \times k}$ be a matrix consisting of independent mean zero variance $\eps$ Gaussian entries (i.e. $S_{i,j} \sim \mathcal{N}(0,\eps)$). It is standard that $\expecf{}{\|S x\|^2 }  = \|x\|^2$ for any vector $x \in \R^k$.
    Let $x_1,\ldots,x_n \in \R^k$ witness the value of $\OPT_k$.
    We analyze
    \[
    \OPT_{1/\eps} \leq \expecf{S}{ \, \frac 1 {n^2} \sum_{i \neq j} \Paren{ 1 - \frac{\|Sx_i - Sx_j\|}{d_{i,j}}}^2 } = \frac 1 {n^2} \sum_{i \neq j} 1 - 2 \frac{ \expecf{S}{\|S x_i - S x_j\| } }{d_{i,j}} + \frac{\|x_i - x_j\|^2}{d_{i,j}^2} \, .
    \]
By Lemma \ref{lem:Gauss-Exp}, we have $\expecf{S}{ \| S x_i - S x_j \| } \geq \|x_i - x_j \| (1 - O(\eps))$. %\ainesh{why is this claim true? this is not standard JL, we have to analyze the sqrt...}
    Therefore, we get 

\begin{equation}
\label{eqn:expanding-opt-1/eps}
    \OPT_{1/\eps} \leq \OPT_{k} + \bigO{\epsilon}\cdot    \frac 1 {n^2}  \sum_{i \neq j} \frac{ \|x_i - x_j\|}{d_{i,j}} \, . 
\end{equation}
    To bound the last term, we observe that by Jensen's inequality,
    \begin{equation*}
    \begin{split}
        1 \geq \OPT_k & = \expecf{i,j}{\Paren{1 - \frac{\norm{ x_i - x_j } }{d_{i,j}} }^2  } \\
        & \geq  1 + \Paren{ \expecf{i,j}{ \frac{\norm{x_i - x_j} }{d_{i,j} } } }^2  - 2 \expecf{i,j}{ \frac{\norm{x_i - x_j} }{d_{i,j} } } 
    \end{split}
    \end{equation*}
Rearranging, we have
\begin{equation*}
    \Paren{ \expecf{i,j}{ \frac{\norm{x_i - x_j} }{d_{i,j} } } }  \leq 2, 
\end{equation*}
and substituting back into \cref{eqn:expanding-opt-1/eps} completes the proof.
\end{proof}

\begin{lemma}[Sketching dimension for expected euclidean norm]\label{lem:Gauss-Exp}
   Let $S \in \R^{ k  \times n}$ be a matrix consisting of i.i.d. Gaussian entries drawn from $\mathcal{N}(0,1/k)$. Then for any vector $x \in \R^n$, we have
   \[\expecf{S}{ \Norm{ S x } } = \Paren{ 1 \pm  \bigO{\frac{1}{k} } } \norm{x} . \] 
\end{lemma}

\begin{proof} 
By scaling, it will suffice to prove the result for unit vectors $x$. In this case, the value $\| S x \|$ is distributed as the square root of a chi-squared variable with $k$ degrees of freedom, which is also known as a chi-distribution with $k$ degrees of freedom. It is known that $\mathbb{E}[\| S x \|] = \sqrt{\frac{2}{k}} \cdot \frac{\Gamma(\frac{k+1}{2})}{\Gamma(\frac{k}{2})}$  (see e.g. Appendix E, p.972 of \cite{gooch2010encyclopedic}). Setting $m=k+1$, we have
   \begin{equation}\label{eqn:first-ex-bound}
        \mathbb{E}[\| S x \|] = \sqrt{\frac{2}{m-1}} \cdot \frac{\Gamma(m/2)}{\Gamma((m-1)/2)}
   \end{equation}
By the the Legendre Duplication Formula~\cite{weisstein2019legendre}, we can write
  \begin{equation}\label{eqn:second-ex-bound}
  \Gamma((m-1)/2) = \sqrt{\pi} \cdot 2^{2-m}\cdot \frac{  \Gamma (m-1)}{\Gamma(m/2) }
    \end{equation}
Plugging \cref{eqn:second-ex-bound} into \cref{eqn:first-ex-bound} yields 
\[ \mathbb{E}[\| S x \|] = \sqrt{\frac{2}{\pi (m-1)}} \cdot2^{m-2} \frac{(\Gamma(m/2))^2}{\Gamma(m-1)}\] 
We can then apply Stirling's approximation $\Gamma(z) = \sqrt{2 \pi } (z-1)^{z-1+1/2} e^{-(z-1)} (1+O(1/z))$ to approximate the Gamma functions above, which yields
\begin{equation*}
    \begin{split}
     \mathbb{E}[\| S x \|]  &= \sqrt{\frac{2}{\pi (m-1)}} \cdot 2^{m-2} \cdot \frac{  \left( \sqrt{2 \pi}(\frac{m}{2}-1)^{\frac{m}{2}-1+\frac{1}{2}}e^{-(\frac{m}{2}-1)}\cdot\left(1+O\left(\frac{1}{m}\right) \right)\right)^2}{ \sqrt{2\pi}(m-2)^{m-2+\frac{1}{2}}e^{-(m-2)}\cdot \left(1+O\left(\frac{1}{m}\right)\right)} \\       
             &= 2 \sqrt{\frac{1}{(m-1)}} \cdot 2^{m-2} \cdot \frac{ (\frac{m}{2}-1)^{m - 1}e^{-(m-2)}\cdot\left(1+O\left(\frac{1}{m}\right) \right)}{  2^{m-1}(m-2)^{-1/2} (\frac{m}{2}-1)^{m-1}  e^{-(m-2)}\cdot \left(1+O\left(\frac{1}{m}\right)\right)}\\
           &=\frac{1}{\sqrt{m-1}}\cdot  \sqrt{m-2} \cdot \left(1+O\left(\frac{1}{m}\right)\right) \\
&= \frac{1}{\sqrt{m-1}} \cdot \sqrt{m-1} \cdot \left(1 - \frac{1}{m-1} \right)^{1/2} \cdot \left(1+O\left(\frac{1}{m}\right)\right) \\
&= 1 \pm O\left(\frac{1}{k}\right)  \\
    \end{split}
\end{equation*}
Where in the last line we used that $(1-x)^{1/2} = (1-\Theta(x))$ for sufficiently small $x \in [0,1]$. 

\end{proof}

\section{Implementation with poly(n)-size LP}
\label{sec:smaller-LP}

%\ainesh{some eyes on this section would be nice}
In this section we show that our algorithm can be implemented with a linear program of size $n^2 \cdot \Paren{\Delta}^{ \tilde{\mathcal{O}}\Paren{ k^2 \log(\Delta)^{k/2}/\eps^{k/2} } } $, resulting in a running time of $n^{\mathcal{O}(1)} \Paren{ \Delta }^{ \tilde{\mathcal{O}}\Paren{ k^2 \log(\Delta)^{k/2}/\eps^{k/2} } }$.

\paragraph{Subsets of the Sherali-Adams LP.}
We can define an LP analogous to the Sherali-Adams LP (see \cref{sec:prelims}) for any subset $\calS \subseteq \binom{n}{t}$, as opposed to all subsets of size $t$. We can introduce variables corresponding to local distributions $\mu_T$ only for $T \in \calS$ and add the consistency constraints such that the marginal distribution on the intersection of any pair $T_1, T_2 \in \calS$ is the same. It is easy to verify that this LP has at most $|\Omega|^{O(t)} \cdot |\calS |^2$ variables and constraints.

\begin{corollary}[Sparsifying the Sherali-Adams LP]
\label{cor:faster-mds}
Given an instance of MDS with aspect ratio $\Delta$, target dimension $k$, and target accuracy $\eps$, the relaxation and rounding introduced in \cref{algo:efficient-algo} can be implemented in $n^{\mathcal{O}(1)} \Paren{ \Delta/\eps }^{\mathcal{O}(\tau)}$ time, where
\begin{equation*}
    \tau = 
    \begin{cases}
    \bigO{  \Paren{ \log\Paren{\log(\Delta)/\eps} \log\Paren{\Delta/\eps} } /\eps   }  & \textrm{ if } k = 1 \\
    \bigO{  ( \log^2\Paren{\log(\Delta)/\eps} \log\Paren{\Delta/\eps} ) /\eps   }  & \textrm{ if } k = 2 \\
    \bigO{ \Paren{ k \log(\log(\Delta)/\eps)\log(\Delta/\eps)^{k/2} } /\eps^{k/2} } & \textrm{ otherwise.} 
    \end{cases}
\end{equation*}
The approximation guarantees remain identical to \cref{thm:efficient-algo-mds}.
\end{corollary}

\begin{proof}
Let $\calS_0 \subset \binom{n}{\tau -2}$ be a  subset of size $\bigO{n/\eps}$, where each element is picked uniformly at random. We augment each tuple in $\calS_0$ of size $\tau-2$ by adding all pairs:

\begin{equation*}
    \calS = \Set{ T\cup \Set{i,j} \textrm{ for } T \in \calS_0 \textrm{ and } i,j\in[n] }.
\end{equation*}
We can now write down an LP relaxation analogous to Sherali-Adams, where we define variables $\mu_T(x)$ for each $T \in \calS$ and consistency constraints for the intersection of every pair $T_1, T_2 \in \calS$. Observe, for the complete Sherali-Adams hierarchy, we can condition on $x_T$ corresponding to any set $T \in \binom{n}{\tau}$ of size $\tau$ and \cref{algo:efficient-algo} conditions on a uniformly random set of size $\tau$. 

In our modified LP, we can condition on $x_T$ corresponding to any set $T \in \calS_0$ and the resulting conditioned pseudo-distribution $\mu_T$ would correspond to a level-$2$ Sherali-Adams program. The bulk of our analysis remains unchanged except for the proof of \cref{lem:var-redux-quantiles}. It suffices to show that for the sparse set $\calS$, conditioning on a random $T$ still hits an index in the $\eps$-th quantile of every pair, and thus reduces the variance. 

To this end, we show that $T$ intersects with any subset of size $\eps n$ with non-trivial probability.  Let $Q \subset [n]$ be a fixed subset of size $\eps \cdot n$. Then, using standard concentration bounds for Bernoulli random variables, 
\begin{equation*}
\begin{split}
    & \probf{\calS_0 }{ \Abs{  \probf{T\sim \calS_0}{ T \cap Q \neq \phi  }   - \expecf{}{\probf{T\sim \calS_0}{ T \cap Q \neq \phi  } } } \geq \eps } \\
    & = \probf{\calS_0 }{ \Abs{  \probf{T\sim \calS_0}{ T \cap Q \neq \phi  }   - \probf{T\sim \binom{n}{\tau-2} }{T \cap Q \neq \phi  }  } \geq \eps  }   \\
    & \leq \delta
\end{split}
\end{equation*}
whenever $\calS_0 \geq \log(1/\delta)/\eps^2$. Recall, there are $\binom{n}{\eps n} \leq 2^{\mathcal{O}\Paren{ \eps n} }$ sets $Q$ of size $\eps n$.
Union bounding over the choice of $Q$, setting $\calS_0$ to have $\bigO{n/\eps}$ random sets of size $\tau-2$ suffices to recover the quantile statements. Therefore, the running time is now 

\begin{equation*}
    \abs{\calS }^{\mathcal{O}(1)} \cdot \abs{\Omega}^{\mathcal{O}(\tau)} =  n^{ \mathcal{O}(1)} \cdot \Paren{ \Delta/\eps}^{ \mathcal{O}(k\cdot \tau) },
\end{equation*}
which concludes the proof.  
\end{proof}

\section{Additional Related Work}
\label{sec:additional-related-work}
We review some additional related work in the broad area of metric embeddings.
For embedding into $\R^k$ the multiplicative
distortion of a pair of input points $i,j$ with input distance 
$d_{ij}$ and output distance $d'_{ij}$, would be $\max \Paren{ d'_{ij} / d_{ij} , d_{ij}/d'_{ij}},$
% For embedding into $\R^k$, the focus would be on optimizing the
% multiplicative distortion
% \[\max \Paren{ \frac{d'_{ij}}{d_{ij}}, \frac{d_{ij}}{d'_{ij}}} \,,
% \]
and the additive distortion, for some objective norm $p$, $\E_{i,j} ||d'_{ij} - d_{ij}||^p .$

In this context, we distinguish two lines of  works: (1) on absolute
distortion, where the goal is to show a bound $X$ such that for any input the distortion is at most $X$; versus (2) on relative or ``data-dependent'' distortion,  where the goal is to algorithmically compute an embedding that for any
input is as close to the best embedding as possible for that input.

In the first line of work, the classic results of Bartal~\cite{DBLP:conf/focs/Bartal96} and Fakcharoenphol, Rao and Talwar~\cite{8ac-DBLP:journals/jcss/FakcharoenpholRT04} show that
one can probabilistically embed any metric space into a tree with 
polylogarithmic distortion. For the case of embedding into Euclidean
spaces and minimizing the multiplicative distortion, we refer to the
work of Matousek~\cite{matouvsek1990bi,matouvsek1996distortion}.
For embedding ultrametrics into Euclidean spaces while minimizing 
the multiplicative distortion, Onak and Sidiropoulos~\cite{onak2008circular} provide a polylog$\Delta$-approximation where (see also~\cite{badoiu2006embedding}). 

    Additionally, there has been a line of work studying various notions of average-distortion when the mapping is constrained to be either non-contracting or non-expanding. For instance, Rabinovich~\cite{rabinovich2008average} considers embedding an arbitrary $n$-point metric $(X,d_X)$ into the line, where the average distance is defined as $\frac{\mathbb{E}_{x,y \in X}[d_X(x,y)]}{\mathbb{E}_{x,y \in X}[|f(x) - f(y)|]}$ subject to the constraint that the embedding  $f:X \to \R$ is non-expanding; the authors also consider the multiplicative inverse of this objective taken over non-contracting embeddings, and argue that the former is a more natural variant due to the phenomenon of concentration of measure. Notice that these objectives capture a notion of how the average distance must decrease when mapping to a lower dimensional metric space, rather than measuring the average distortion that a pair of points incurs under this mapping (which is the goal of MDS).  A related objective, studied in \cite{lee2005metric,bartal2019dimensionality}), which is closer to the spirit of MDS, is $\mathbb{E}[|x- y|/d(x,y)]$ taken over non-contracting embeddings.

\paragraph{``Data''-dependent embeddings}
In this paper, we are concerned with the second line of work.
Strong inapproximability results for max-distortion in the mid 2000s (e.g. \cite{DBLP:conf/focs/MatousekS08}) have led to some efforts to design approximation algorithms for average-case notions of distortion, among them MDS. Here ``average-case'' is in the sense that we take the average over $i,j$ pairs, not that $d_{ij}$s are random.
For the additive distortion and the problem of embedding into the line (i.e.: the $d'$ distance should come from a line metric),
the problem is NP-Hard~\cite{haastad1998fitting,ivansson2000computational} and Hastad et al.~\cite{haastad1998fitting} showed provided a 2-approximation algorithm that was later improved by Badoiu to the 
plane under the $\ell_1$-metric~\cite{DBLP:conf/soda/Badoiu03}.
Dhamdhere \cite{6ac-DBLP:conf/approx/Dhamdhere04} considered the problem of embedding into a line metric in a non-contractive way to minimize additive distortion and presented an $\calO(\log n)$ approximation. In the context where the objective is the multiplicative distortion \cite{doi:10.1137/17M1113527} provides an
$O(\sqrt{n})$-approximation to the problem and ~\cite{badoiu2005low}
later provided an $O(\Delta^{3/4})$-approximation.

There have been few provable bounds on approximating average-case distortion style objectives for dimensionality-reduction problems.  Bartal, Fandina and Neiman~\cite{bartal2019dimensionality} give provable approximation guarantees for a several variants of multi-dimensional scaling, the closest of which to our setting,  i.e. the Kamada-Kawai objective, is the objective 
$$\expecf{i\neq j \in [n]}{ \Paren{ 1 - \max\Paren{ \frac{\|x_i - x_j\|}{d_{i,j}}, \frac{d_{i,j}}{\|x_i-x_j\|} } }^2 }.$$ 
They first employ an embedding of the $d_{i,j}$'s into high dimensional Euclidean space $\R^n$ by solving the standard semi-definite program. Next, they apply Johnson-Lindenstrauss (JL) projections to reduce the dimension from $n$ to any target dimension $k$. The main contribution of \cite{bartal2019dimensionality} is in the moment based analysis of JL. However, it is easy to see that JL will result unavoidably result in fixed additive error that is a function only of $k$; Specifically, the authors obtain an additive $O(1/\sqrt{k})$ approximation to several objectives, including the aforementioned. However, since $\OPT \in [0,1]$, this guarantee is non-vacuous only for target dimension larger than $1/\eps^2$, regardless of running time. Additionally, their result requires the initial distances $\{d_{i,j}\}$ to be a metric---something that we do not assume. In contrast, our results give non-trivial approximations for any $k$, for which the additive $\eps$ error can be made arbitrarily small at the cost of increasing the runtime of the algorithm.

The literature on embeddings into tree metrics or ultrametric is also
very rich. Motivated by biology applications, researchers have  
studied reconstruction of phylogenies as a tree or ultrametric
embedding problem: given pairwise dissimilarities between $n$
elements, the goal is to fit the elements to a tree or ultrametric while minimizing the maximum additive distortion 
or the sum over all pairs of elements of the additive distortion.
Farach-Colton, Kannan and Warnow~\cite{6-DBLP:journals/algorithmica/FarachKW95} showed
that for the ultrametric case an optimal solution can be found in $\Theta(|S|^2)$ time (see also~\cite{Cohen-AddadSL20,Cohen-AddadJL21} for when the input is Euclidean). For embedding into a tree, Agarwala, Bafna, Farach-Colton, Paterson and Thorup~\cite{DBLP:journals/siamcomp/AgarwalaBFPT99} provided a constant factor approximation and showed the problem is APX-Hard.
%\mtcom{Is this not related to what Farach and Kannan did, and if so, should we not refer back?}
%\nnote{Aren't we repeating ourselves here? This was mentioned before.}
 If the input distances have at most $k$ distinct values, then Harb, Kannan and McGregor \cite{12ac-DBLP:conf/approx/HarbKM05} showed a factor
$\calO(\min\{n,k\log n\}^{1/p})$ approximation.
The best bounds known for the ultrametric variant of the problem are due to Cohen-Addad et al.~\cite{DBLP:conf/focs/Cohen-Addad0KPT21} and Ailon and Charikar \cite{DBLP:journals/siamcomp/AilonC11}.
This has also been considered in the context of stochastic inputs (see Farach-Colton and 
Kannan~\cite{DBLP:journals/jacm/FarachK99},  Mossel and Roch~\cite{16ac-DBLP:conf/stoc/MosselR05}, Henzinger, King, Warnow~\cite{DBLP:journals/algorithmica/HenzingerKW99} and the references therein).

% \section*{Acknowledgements}
% We thank Moses Charikar, Piotr Indyk, and Robert Kleinberg for helpful conversations.
% SBH was funded by NSF CAREER award no. 2238080 and MLA@CSAIL.
% AB was funded by the Foundations of Data Science Institute, NSF award no. 2023505. 

\section{MDS Versus Principal Component Analysis}
\label{sec:MDSvPCA}
We now briefly describe an instance, with small Kamada-Kawai objective cost, where the optimal Kamada-Kawai objective produces a qualitatively more interesting output than that obtained by PCA. In particular, we show that Kamada-Kawai better captures cluster structure in this instance. The instance itself is one that already has good low-dimensional structure. Since PCA totally projects away any information in the dataset which is orthogonal to the principle component, this information is totally lost by PCA. On the other hand, KK being a non-linear transformation, this orthogonal information can still captured in an optimal embedding.

Consider the input dataset $P \subset \R^2$, defined by $k$ copies of the point $(i,0) \in \R^2$ for each $i=\{1,2,3,\dots,n\}$, for a sufficiently large constant $k \geq 1$, along with the set of points $\cup_{i \in [n]} \{(i,\eps), (i, -\eps)\}$ for a sufficiently small constant $\eps$. In other words, $P$ consists of $n$ clusters, each separated by distance $1 + O(\eps)$, where each cluster consists of $k+2$ points centered at $(i,0)$ and symmetric about the $x$-axis. This input example is illustrated in Figure \ref{fig:PCAexample}. Our goal will be to embed $P$ into the line $\R^1$.

Now observe that the principle component of the dataset $P$ is in direction $v=(1,0)$, thus PCA would project $P$ onto $v$. Thus, the optimal PCA maps $(i,\pm \eps) \mapsto i$ and $(i,0) \mapsto i$ for each $i \in [n]$, thereby collapsing the cluster structure. On the other hand, notice that this projection would have KK cost at least $2nk(1+O(\eps))$, since \textbf{(1)} each of the points $(i,\pm \eps)$ would pay a cost of $1 = (1-\frac{0}{\eps})^2$ in the KK objective for having their distance to $(i,0)$ change from $\eps$ to $0$, as well as a cost of at most $(2k+1)\sum_{i\geq 1} (1-\frac{i}{i + O(\eps)})^2 = O(\eps k)$ for distorting distances to the points with $x$-coordinate not equal to $i$ (i.e. the distance between $(i,\eps)$ and $(j,0),(j,-\eps)$ is distorted by a multiplicative $(1+\eps)$ factor under this mapping for $i \neq j$).

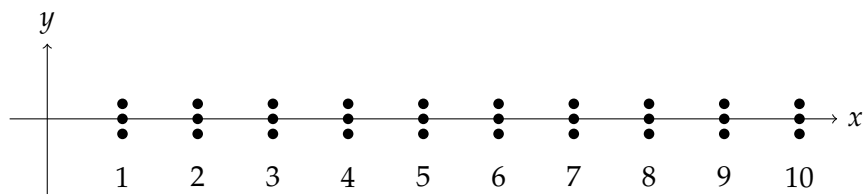
\begin{figure}[!h]
\begin{center}
    
\begin{tikzpicture}
% Parameters
\def\n{10}  % Number of points on the x-axis
\def\eps{0.2}  % Vertical offset

% Axes
\draw[->] (-0.5,0) -- (\n+0.5,0) node[right] {$x$};
\draw[->] (0,-1) -- (0,1) node[above] {$y$};

% Labels for points in the middle
\foreach \i in {1,...,\n} {
   \node[below] at (\i,-.5) {$\i$}; 
}

% The actual points
\foreach \i in {1,...,\n} {
    \fill (\i,0) circle (2pt); 
    \fill (\i,\eps) circle (2pt);
    \fill (\i,-\eps) circle (2pt);
}

\end{tikzpicture}

\end{center}

\caption{Illustration of an input instance where MDS preserves meaningful cluster structure but PCA does not. }\label{fig:PCAexample}
\end{figure}
\begin{figure}[!h]
    \centering
   
\begin{tikzpicture}

% Parameters
\def\n{10}  % Number of points on the line 
\def\eps{0.2}  % Small offset value

% Setting the x-axis limits a little wider for better visualization
\def\xaxislim{\n+1.5} 

% Axes
\draw[->] (0,0) -- (\xaxislim,0) node[right] {$x$};

% Black points at integer positions
\foreach \i in {1,...,\n} {
   \fill (\i,0) circle (2pt) [black]; 
}

% Red points with offset
\foreach \i in {1,...,\n} {
    \fill (\i+\eps,0) circle (2pt) [red];
    \fill (\i-\eps,0) circle (2pt) [red];
}

% Labels for points
\foreach \i in {1,...,\n} {
   \node[below] at (\i,0) {$\i$}; 
}

\end{tikzpicture}
    \caption{Illustration of the optimal MDS mapping for this instance. The red points are the points previously at $(i, \pm \eps)$ that which were mapped onto the line by MDS, whereas PCA would collapse each pair of red points onto their central black point.}
    \label{fig:MDSOptEmbedding}
\end{figure}
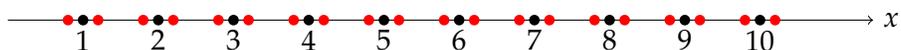

Instead of this embedding, consider the MDS embedding defined by 

\[ (x,y) \mapsto \begin{cases}
    i & \text{ if } (x,y) = (i,0) \\
     i + \eps & \text{ if } (x,y) = (i,\eps)\\ 
      i - \eps & \text{ if } (x,y) = (i,-\eps)
\end{cases}\]
An illustration of this embedding is given in Figure \ref{fig:MDSOptEmbedding}. 
In this case, not that the distances between points in the same cluster remain unchanged, and the MDS cost for a given point $(i,\pm \eps)$ is at most $(2k+2)\sum_{i\geq 1} (1-\frac{i}{i \pm \Theta(\eps)})^2 = O(\eps k)$, again for distorting the distances to points whose original $x$-coordinate was not equal to $i$. Thus the total (unnormalized) KK cost of this embedding is $O(\eps n k)$, as opposed to the cost of $O(n k)$ that would be incured from the PCA embedding.

\section*{Acknowledgements}
We thank Moses Charikar, Piotr Indyk, and Robert Kleinberg for helpful conversations.
SBH was funded by NSF CAREER award no. 2238080 and MLA@CSAIL.
AB was funded by the Foundations of Data Science Institute, NSF award no. 2023505.

\newpage 

\bibliographystyle{alpha}
\bibliography{bibliography}

% \newpage
% \appendix

%\input{structure-theorem-recursive-partition}
%\input{quantile-bounds-with-randomized-partition}

%\newpage
%\begin{lemma}
%\label{lem:quantiles-with-subsets}Let $\mu$ be a pseudo-distribution minimizing the objective in Algorithm~\ref{algo:efficient-algo}, and suppose that $\min_{i,j} d_{i,j} \geq 1$. 
%Given $0< \eta, \delta <1$, then for all but  $\bigO{ (\delta \log(\Delta) + \OPT )\cdot n^2 }$ pairs $(i,j)$, we have 
%\begin{equation*}
 %\calQ\Paren{ \Set{\pexpecf{\mu}{ \norm{x_i - x_k }}}_{k \in  [n] }, \eta  }  + \calQ\Paren{ \Set{\pexpecf{\mu}{ \norm{x_j - x_k }}}_{k \in  [n] }, \eta }  \leq \bigO{ \Paren{ \frac{ \eta }{\delta} }^{1/k} } \pexpecf{\mu}{ \norm{x_i - x_j }}
%\end{equation*}
%\end{lemma}

\end{document}